\documentclass[pra,twocolumn,floatfix,notitlepage,superscriptaddress]{revtex4-2}                                                                              
\usepackage{mathtools}
\usepackage{scalerel,stackengine}
\stackMath
\newcommand\reallywidehat[1]{%
\savestack{\tmpbox}{\stretchto{%
  \scaleto{%
    \scalerel*[\widthof{\ensuremath{#1}}]{\kern-.6pt\bigwedge\kern-.6pt}%
    {\rule[-\textheight/2]{1ex}{\textheight}}%WIDTH-LIMITED BIG WEDGE
  }{\textheight}% 
}{0.5ex}}%
\stackon[1pt]{#1}{\tmpbox}%
}
\usepackage{amsmath}
\usepackage{ytableau}
\usepackage{amssymb,amsmath,amstext,amsthm}
\usepackage{graphicx}
\usepackage{epstopdf}
\usepackage{color}
\usepackage{dsfont}
\usepackage{bm}
\usepackage{appendix}
\usepackage[T1]{fontenc}
\usepackage{bbm}
\usepackage{latexsym}
\usepackage{bbm}
\usepackage{float}
\usepackage{verbatim}
\usepackage{lipsum}
\usepackage{braket}
\usepackage{mathtools}
\usepackage{mathrsfs}
\usepackage{dcolumn}
\usepackage{xcolor}   
\usepackage{bm}        
\usepackage{amssymb}   % for math
\usepackage{physics}
\usepackage{tikz}
\usepackage{gensymb}
\usepackage{xcolor}
\usepackage{float}
\usepackage{dsfont}
\usepackage[caption = false]{subfig}
\usepackage{mathtools}
 \usepackage{amsthm}
 \usepackage[english]{babel}
 \usepackage[bookmarks,bookmarksopen,bookmarksdepth=2]{hyperref}
\hypersetup{
    colorlinks=true,
    citecolor=blue,
    linkcolor=blue,
    filecolor=magenta,
    urlcolor=black
}
 \usepackage{quantikz}
\usepackage[capitalize]{cleveref}

\newtheorem{theorem}{Theorem}
\newtheorem{lemma}{Lemma}

\def\frac#1#2{{\begingroup #1\endgroup\over #2}}

\renewcommand{\geq}{\geqslant}

\def \addLANLT {Theoretical Division, Los Alamos National Laboratory, Los Alamos, NM 87545, USA}

\def \addLANLS {Statistics Group, Los Alamos National Laboratory, Los Alamos, NM 87545, USA}

\begin{document}

\title{Length scale estimation of excited quantum oscillators}

\author{Tyler Volkoff}
\email{volkoff@lanl.gov}
\affiliation{\addLANLT}
\author{Giri Gopalan}
\email{ggopalan@lanl.gov}
\affiliation{\addLANLS}

\begin{abstract}
Massive quantum oscillators are finding increasing applications in proposals for high-precision quantum sensors and interferometric detection of weak forces. Although optimal estimation of certain properties of massive quantum oscillators such as phase shifts and displacements have strict counterparts in the theory of quantum estimation of the electromagnetic field, the phase space anisotropy of the massive oscillator is characterized by a length scale parameter that is an independent target for quantum estimation methods. We show that displaced squeezed states and excited eigenstates of a massive oscillator exhibit Heisenberg scaling of the quantum Fisher information for the length scale with respect to excitation number, and discuss asymptotically unbiased and efficient estimation allowing to achieve the predicted sensitivity. We construct a sequence of entangled states of two massive oscillators that provides a boost in length scale sensitivity equivalent to appending a third massive oscillator to a non-entangled system, and a state of $N$ oscillators exhibiting Heisenberg scaling with respect to the total energy.
\end{abstract}
\maketitle

\section{Introduction}
Massive quantum oscillators provide a framework for describing a wide range of natural and engineered particle systems:  from shell models of the atomic nucleus \cite{MOSHINSKY1962384}, to atoms in a harmonic optical trap \cite{lieb} and ions in radio-frequency traps \cite{winewine}, to optically-controlled nano- and micro-mechanical resonators \cite{Groblacher2009,Samanta2023,PhysRevA.89.013854,Rips_2012,Banerjee:23,choi2024quantumlimitedimagingnanomechanical,TANG2022127966,10.1116/5.0022349}. As methods for controlling the quantum state of massive oscillators have progressed, they increasingly appear in technological applications, such as proposals for boosting quantum metrology \cite{Woolley_2008,rahman2024} and gravitational wave detection \cite{osti_5123931,lll,PhysRevLett.105.070403}. From the perspective of quantum information theory, massive quantum oscillators provide examples of continuous-variable quantum systems, so the well-developed theories of computation, communication, and estimation in the continuous-variable quantum setting carry over, at least in principle, to the example of massive quantum oscillators.

However, unlike massless quantum oscillators such as the quantum electromagnetic field, massive quantum oscillators are characterized by anisotropic phase space distribution in the ground state. From a kinematic viewpoint, this anisotropy amounts to a different choice of complex structure for the symplectic space $\mathbb{R}^{2}$, with $q=\sqrt{\hbar}{a+a^{\dagger}\over \sqrt{2}}$, $p=\sqrt{\hbar}{ia^{\dagger}-ia\over \sqrt{2}}$ giving the canonical phase and amplitude quadratures of a single mode of the electromagnetic field \cite{Mandel_Wolf_1995} and $q=\mathcal{L}\sqrt{\hbar}\left( {a+a^{\dagger}\over \sqrt{2}} \right)$, $p={\sqrt{\hbar}\over \mathcal{L}} \left( {ia^{\dagger}-ia\over \sqrt{2}} \right)$ giving the position and momentum operators of the massive oscillator for some $\mathcal{L}>0$ with $\mathcal{L}\sqrt{\hbar}$ having units of length. Specifically,
\begin{align}
    \Delta &:= \begin{pmatrix} 0&1 \\ -1 & 0\end{pmatrix} \nonumber \\
    J&:= \begin{pmatrix} 0&-\mathcal{L}^{-1} \\ \mathcal{L} & 0\end{pmatrix}
\end{align}
give the standard symplectic form and an operator of complex structure, respectively, with $J$ clearly satisfying the defining condition $J^{2}=-I$ and the compatibility conditions $\Delta J \succeq 0$, $J^{\intercal}\Delta J = \Delta$.

In practice, this complex structure has the important physical consequence of scaling the contributions of the kinetic and potential energies relative to the total energy of the oscillator. These relative contributions can be phrased in terms of the length scale of the oscillator, i.e.,  the parameter $L:=\sqrt{ \hbar \over m\omega}$, with $m$ the oscillator mass, $\omega$ the oscillator frequency, and the constant $\hbar$ can be viewed as arising from the canonical commutation relation $[q,p]=i\hbar$. The physical importance of $L$ stems from the fact that it defines the width of the normal distribution characterizing the position statistic of the lowest energy state of the oscillator. In the present work, we analyze the problem of quantum estimation of $L$ (for convenience, we will consider estimation of $d:=L^{-2}$ so that the parameter of interest appears in the normalization factors of the relevant wavefunctions and is monotonically increasing with the oscillator mass $m$).

The problem of estimation of the length scale $d$ given the state $\ket{\psi_{d}}$ of a massive oscillator is not the same as the continuous-variable squeezing estimation problem. The latter problem falls into the shift model of quantum estimation \cite{holevo}, in which the parametrized states have the form $\ket{\psi_{r}}=S(r)\ket{\psi}$, where $S(r)$ is the unitary squeezing operator, and has been extensively studied \cite{PhysRevA.73.062103}. By contrast, parametrized states of the length scale estimation problem take the form
\begin{equation}
    \ket{\psi(d)}=\sum_{n=0}^{\infty}c_{n}\ket{\psi_{n}(d)} \; , \; \sum_{n=0}^{\infty}\vert c_{n}\vert^{2}=1,
    \label{eqn:states}
\end{equation}
where 
\begin{equation}
\langle q=x \vert \psi_{n}(d)\rangle:= \left( d\over \pi\right)^{1/4}{1\over \sqrt{2^{n}n!}} e^{-{d\over 2}x^{2}}H_{n}(x\sqrt{d}) 
\label{eqn:wfdef}
\end{equation}
is the wavefunction of the $n$-th excited state of a massive quantum oscillator \cite{Sakurai1993Modern}. In (\ref{eqn:wfdef}), the canonical quadrature $q$ is related to the excitation creation and annihilation operators according to $q={a+a^{\dagger}\over \sqrt{2d}}$ and $H_{n}(x)$ is the $n$-th Hermite polynomial. We use $\psi_{n}(d)$ to denote the rank-one projection $\ket{\psi_{n}(d)}\bra{\psi_{n}(d)}$.

\section{Preliminaries}

In this section we state the preliminaries from quantum estimation theory.  Our approach to the problem of estimation of $d$ is based on the existence of a normalized positive operator-valued measure (POVM) $\lbrace E(\hat{d})\rbrace$ on the half-line $(0,\infty)$ of estimators $\hat{d}$ that saturates the following single-shot quantum Cram\'{e}r-Rao inequality \cite{holevo,PhysRevLett.72.3439}
\begin{align}
    \min_{\lbrace E(\,\hat{d}\,)\rbrace \text{POVM}}\int_{0}^{\infty} \, \left( \hat{d}-d \right)^{2}\text{tr}[\rho_{d}E(\hat{d})] &\ge {1\over \text{QFI}(d)} 
    \label{eqn:qcrb}
\end{align}
which is valid for an unbiased estimator $\hat{d}$ if the probe state $\rho_{d}$ is differentiable with respect to $L$ and satisfies a further regularity condition (Section 6.5 of \cite{holevo}). These requirements guarantee the existence of the symmetric logarithmic derivative (SLD) $L_{d}$, which is a self-adjoint operator that satisfies
\begin{equation}
    \partial_{d}\rho_{d} = {1\over 2}L_{d}\rho_{d} + {1\over 2}\rho_{d}L_{d}.
    \label{eqn:sylv}
\end{equation}
The quantum Fisher information (QFI) is defined by
\begin{equation}
    \text{QFI}(d) := \text{tr}\rho_{d}L_{d}^{2}.
    \label{eqn:qfqf}
\end{equation}
The SLD is not unique in general due to the fact that the right hand side of (\ref{eqn:sylv}) involves only the support of $\rho_{d}$ \cite{PRXQuantum.4.040336}. In the present work, we will use the definition (\ref{eqn:qfqf}) to calculate the QFI, and will also use the relation
\begin{equation}
    \text{QFI}(d) = -4\partial_{d'}^{2}\sqrt{F(\rho_{d},\rho_{d'})}\big\vert_{d'=d}
    \label{eqn:bb}
\end{equation}
relating the QFI to the Bures metric on the state manifold, the latter being defined in terms of the quantum fidelity $F(\rho_{d},\rho_{d'}):= \left(\text{tr}\sqrt{\sqrt{\rho_{d'}}\rho_{d}\sqrt{\rho_{d'}}}\right)^{2}$.

In practice, the optimal POVM in (\ref{eqn:qcrb}) often has a structure that is not easily related to an observable that can be measured in an experiment. If, instead of the optimal observable $\int_{0}^{\infty}\hat{d}E(\hat{d})$, one measures another observable $A$ and carries out classical postprocessing to obtain an estimator $\hat{d}$, one finds that the classical Fisher information (CFI) of the probability distribution $p_{d}^{(A)}(\hat{d})$ arising from measurement of $A$ satisfies \cite{holmom}
\begin{equation}
    {\left( \partial_{d}\langle A \rangle_{\rho_{d}}\right)^{2} \over \text{Var}_{\rho_{d}}A}\le \text{CFI}^{(A)}(d)
    \label{eqn:snrupp}
\end{equation}
where
\begin{equation}
    \text{CFI}^{(A)}(d):= 4\int_{0}^{\infty} \left( \partial_{d}\sqrt{p_{d}(\hat{d})} \right)^{2}.
    \label{eqn:snrsnr}
\end{equation}
The statistical meaning of the left hand side of (\ref{eqn:snrupp}) is that it defines the variance of asymptotically normally distributed estimates $\hat{d}$ over many measurements of $A$, according to the central limit theorem \cite{ps,RevModPhys.90.035005}. Physically, it can be considered as the reciprocal of the noise of a units-corrected observable and its reciprocal is referred to as ``mean square noise'' in (3.16) of \cite{PhysRevA.33.4033}. We refer to the quantity on the left hand side of (\ref{eqn:snrupp}) as the reciprocal of the mean square noise. Equality is obtained if the estimator from measurement of $A$ is asymptotically efficient.
The proof of (\ref{eqn:qcrb}) proceeds through the fact that $\text{CFI}^{(A)}(d) \le \text{QFI}(d)$ for any self-adjoint $A$, so finding an observable with the reciprocal of the mean square noise having a value near $\text{QFI}(d)$ implies near-optimality of method of moments estimation of $d$ from measurement of $A$ \cite{PhysRevLett.122.090503}.

\section{Benchmarks: classical and Gaussian states of the massive quantum oscillator\label{sec:bench}}

By analogy with the notion of classical states of the electromagnetic field \cite{PhysRev.130.2529}, we consider a classical state of a single massive oscillator to be a convex combination of coherent states, where a coherent state $\ket{\psi_{\alpha}(d)}$ is defined by Poisson excitation statistics, i.e., $c_{n} = e^{-{\vert \alpha \vert^{2}\over 2}}{\alpha^{n}\over \sqrt{n!}}$, $\alpha \in \mathbb{C}$, in (\ref{eqn:states}). The $d$-dependence of the coherent state is evident through its wavefunction (for $\alpha \in \mathbb{R}$)
\begin{equation}
    \langle q=x\vert \psi_{\alpha}(d)\rangle = \left( {d\over \pi} \right)^{1\over 4}e^{-{d\over 2}\left( x-\sqrt{2\over d}\alpha\right)^{2}}.
\end{equation}
The inner product
\begin{equation}
    \langle \psi_{\alpha}(d)\vert \psi_{\alpha}(d')\rangle = (dd')^{{1\over 4}}e^{-2\alpha^{2}}\sqrt{2\over d+d'}e^{\alpha^{2}(\sqrt{d}+\sqrt{d'})^{2}\over d+d'}
\end{equation}
can be used to obtain the QFI according to the formula $\text{QFI}(d)=-4\partial_{d'}^{2}\langle \psi(d)\vert\psi(d')\rangle\vert_{d'=d}$ that follows from (\ref{eqn:bb}) when a pure probe state $\ket{\psi(d)}$ with real wavefunction is utilized. One obtains
\begin{equation}
    \text{QFI}(d) = {1+2\alpha^{2}\over 2d^{2}},
    \label{eqn:cohqfi}
\end{equation}
which scales linearly with the energy of the state $\ket{\psi_{\alpha}(d)}$, i.e., linearly with $\alpha^{2}$. Convexity of the QFI implies that this linear scaling is the best achievable over all classical states of the massive oscillator.

A squeezed state of the massive oscillator is obtained by applying the unitary squeeze operator $S(z)=e^{{1\over 2}(\bar{z}a^{2} - za^{\dagger 2})}$, $z\in \mathbb{C}$, to the vacuum state $\ket{\psi_{0}(d)}$. The squeezed state is Gaussian, and taking $z={1\over 2}\ln D$ for $D>1$, one obtains the wavefunction
\begin{equation}
    \left( {dD\over \pi}\right)^{1\over 4}e^{-{dD\over 2}x^{2}}.
    \label{eqn:sqsq}
\end{equation}
Although it is clear that the position and momentum fluctuations of the squeezed massive oscillator depend on the product $dD$, we consider that the length scale $d$ corresponds to an intrinsic property of the system, whereas the unitless scale factor $D$ is achieved by an external control coupled to the system. Therefore, $d$ and $D$ are independent parameters of the state. However, the generator of the unitary squeezing operator is a linear combination of operators $qp$ and $pq$ and is therefore independent of the mass of the oscillator. This can be contrasted with the generator of phase space displacements $\alpha a^{\dagger} - \bar{\alpha}a$, which depends on the mass through the choice of complex structure, and displaces the oscillator in phase space distance units of $L$ along position axis and ${\hbar \over L}$ along the momentum axis. Therefore, squeezing the ground state of a massive oscillator does not give a probe state that is more sensitive to $d$ compared to the ground state itself.

The fact that phase space displacement increases the sensitivity of the ground state to changes in the oscillator length then motivates one to consider a probe state $\ket{\psi_{\alpha,D}(d)}$ formed by displacing the squeezed ground state in (\ref{eqn:sqsq}). For real $\alpha$ and real $D>1$, the wavefunction is given by
\begin{equation}
    \langle q=x \vert \psi_{\alpha,D}(d)\rangle = \left( {dD\over \pi}\right)^{1\over 4}e^{-{dD\over 2}(x-\sqrt{2\over d}\alpha)^{2}},
    \label{eqn:sqsq2}
\end{equation}
and one computes the QFI for fixed $\alpha$ and $D$ to be
\begin{equation}
    \text{QFI}(d)= {1 + 2\alpha^{2}D \over 2d^{2}}.
\end{equation}
The total energy $E$ of the state $\vert \psi_{\alpha,D}(d)\rangle$ satisfies $E \propto \alpha^{2} + {D^{2}+D^{-2}-2\over 4}$. Remarkably, the problem of maximizing $\alpha^{2}D$ with respect to the total energy constraint $E$ also appears in the optical phase estimation tasks arising in $SU(1,1)$ interferometers \cite{PhysRevA.33.4033}, and $SU(2)$ interferometers \cite{PhysRevD.23.1693}. Optimally allocating energy to displacement and squeezing then results in $O(E^{2})$ scaling of the QFI, i.e., Heisenberg scaling, which is the focus of the next section. It is notable that the parameter $d$ and the oscillator phase $\phi:= \omega t$ (for evolution time $t$) depend linearly on the frequency $\omega$. Therefore, for fixed mass and fixed evolution time, one can see from $\phi = {\hbar td\over m}$ that the optimal sensitivity scaling of phase estimation allows one to infer the optimal sensitivity scaling to both phase shifts and length scale changes. However, one must be careful to note that there are probe states that are sensitive to length scale changes, but do not accumulate a dynamical phase.

\section{Heisenberg scaling\label{sec:hs}}

For a system of massive oscillators, we characterize Heisenberg scaling of the QFI for the length scale parameter $d$ by $\text{QFI(d)}=O(E^{2})$, where $E$ is the expected total energy of the massive oscillator. Similar to the case of rotation estimation, i.e., estimation of the parameter $\theta$ in $e^{i\theta a^{\dagger}a}$, the $O(E^{2})$ scaling can be interpreted as the scaling of the maximal area in phase space that is relevant for distinguishing probe states having close parameter values and each having total energy $E$. However, whereas energy eigenvectors $\ket{\psi_{n}(d)}$ are insensitive to rotations, we finds that these states provide Heisenberg scaling for estimation of $d$. The following Lemma will allow us to obtain an analytical formula for the QFI.
\begin{lemma}
\label{lem:ooo}
    For any $n\in \mathbb{N}_{\ge 2}$,
    \begin{align}
        \partial_{d}\ket{\psi_{n}(d)} &= {1\over 4d}\left( \sqrt{n(n-1)}\ket{\psi_{n-2}(d)} \right. \nonumber \\
        &{} \left. - \sqrt{(n+1)(n+2)}\ket{\psi_{n+2}(d)}\right).
        \label{eqn:slsl}
    \end{align}
    The derivatives of the ground state and first excited state are, respectively,
    \begin{align}
        \partial_{d}\ket{\psi_{0}(d)} &=-{\sqrt{2}\over 4d}\ket{\psi_{2}(d)}\nonumber \\
        \partial_{d}\ket{\psi_{1}(d)} &=-{\sqrt{6}\over 4d} \ket{\psi_{3}(d)}. 
        \label{eqn:sllow}
    \end{align}
\end{lemma}
\begin{proof}
    The derivative $\partial_{d}\langle q=x\vert \psi_{n}(d)\rangle$ can be written as the product of $e^{-{d\over 2}x^{2}}$ and an order-$(n+2)$ polynomial over $x\sqrt{d}$ with $d$-dependent coefficients. This polynomial can then be expanded over the $L^{2}$-orthogonal basis of Hermite polynomials $H_{n}(x\sqrt{d})$ (with respect to measure $e^{-{d\over 2}x^{2}}dx$ on $\mathbb{R}$). The relevant identities for this expansion are:
    \begin{align}
        u^{2}H_{n}(u)&= {1\over 4}H_{n+2}(u)+{2n+1\over 2}H_{n}(u) \nonumber \\
        &+ n(n-1)H_{n-2}(u) \label{eqn:usqh} \\
        u\partial_{u}H_{n}(u)&= nH_{n}(u)+2n(n-1)H_{n-2}(u).
    \end{align}
    Rescaling the individual terms and bringing back the factor $e^{-{d\over 2}x^{2}}$ gives an expansion in terms of the orthonormal massive oscillator wavefunctions in which the only non-zero coefficients are on $\ket{\psi_{n\pm 2}(d)}$, with the values in (\ref{eqn:slsl}).
    The derivatives of the ground state and first excited state are verified directly from the expressions for $H_{n}(x\sqrt{d})$, $n=0,1,2,3$.
\end{proof}
We now show that Heisenberg scaling is obtained for the pure probe state $\ket{\psi_{n}(d)}$.
\begin{theorem}
For $n\in \mathbb{N}_{\ge 0}$, the parametrized state $\ket{\psi_{n}(d)}$ has
    \begin{equation}
        \mathrm{QFI}(d)={n^{2}+n+1\over 2d^{2}} = \mathrm{Var}_{\ket{\psi_{n}(d)}}q^{2} .
        \label{eqn:qfqfqf}
    \end{equation}
    \label{thm:hhh}
\end{theorem}
\begin{proof}
    Since $\ket{\psi_{n}(d)}$ is a pure state, the SLD is $L_{d}:=2\partial_{d}\psi_{n}(d)$. By Lemma \ref{lem:ooo}, the component of $L_{d}^{2}$ on the rank one operator $\psi_{n}(d)$, i.e., $\text{QFI}(d)$, is equal to $4\Vert \partial_{d}\ket{\psi_{n}(d)}\Vert^{2}$. Simplifying by using the explicit expression in (\ref{eqn:slsl}) for $n\ge 2$ and (\ref{eqn:sllow}) for $n=0,1$ gives the left equality in (\ref{eqn:qfqfqf}).

    The right equality is obtained by applying (\ref{eqn:usqh}) once to obtain $\langle \psi_{n}(d)\vert q^{2}\vert \psi_{n}(d)\rangle$ and twice to obtain $\langle \psi_{n}(d)\vert q^{4}\vert \psi_{n}(d)\rangle$. This equality should be viewed as a kinematic relationship rather than a general expression for the QFI, since both the probe state $\ket{\psi_{n}(d)}$ and the position operator $q$ depend on $d$.
\end{proof}

For practical applications, it is desirable to find an observable $A$ such that the reciprocal of the mean square noise satisfies
\begin{equation}
{\left( \partial_{d}\langle A \rangle_{\ket{\psi_{n}(d)}}\right)^{2} \over \text{Var}_{\ket{\psi_{n}(d)}}A} = O(n^{2})
\label{eqn:snsn}
\end{equation}
globally (i.e., for all $d$) because this condition guarantees Heisenberg scaling of the CFI and, therefore, the QFI, according to (\ref{eqn:snrupp}). Consider taking $A=\psi_{0}(1)$, the projection onto the Fock vacuum $\ket{\psi_{0}(1)}$ for the fiducial length scale $d=1$. From (\ref{eqn:sqsq}), one notes that in the basis $\lbrace \ket{\psi_{n}(d)}\rbrace_{n=0}^{\infty}$, the state $\ket{\psi_{0}(1)}$ is an anti-squeezed state
\begin{equation}
\ket{\psi_{0}(1)} = {1\over \sqrt{\cosh \ln d}}\sum_{n=0}^{\infty} \sqrt{{2n\choose n}} \left( {-\tanh\ln d \over 2}\right)^{n}\ket{\psi_{2n}(d)}.
\end{equation}  The derivative  $\partial_{d}\langle A\rangle_{\ket{\psi_{n}(d)}}$ associated with $A$ has two undesirable properties 1. it is zero when the excitation number $n$ of the probe state $\ket{\psi_{n}(d)}$ is odd, and 2. for $n=2m$, one finds the scaling behavior $\partial_{d}\langle A\rangle_{\ket{\psi_{n}(d)}}=O(\sqrt{n})$ and $\text{Var}_{\ket{\psi_{n}(d)}}A=O(1)$ with respect to $n$, so the quadratic scaling in (\ref{eqn:snsn}) cannot hold. However, a projective readout with Heisenberg scaling regardless of the parity of $n$ can be obtained when allowing two copies of the probe state (namely $\ket{\psi_{n}(d)}^{\otimes 2}$) and taking $A$ to be the projection onto the entangled state \begin{equation}e^{{\pi\over 4}(a^{\dagger}_{1}a_{2} - a_{2}^{\dagger}a_{1})}e^{i{\pi\over 2}a^{\dagger}_{2}a_{2}}\ket{\psi_{0}(1)}^{\otimes 2}.\label{eqn:bell}\end{equation}
This projection can be viewed as a (finite-energy) Bell measurement of a system of two massive quantum oscillators. Note that although (\ref{eqn:bell}) is a linear optical transformation of two copies of the $d=1$ Fock vacuum, the linear optical transformation is with respect to the quantum oscillator at general $d$ and therefore generates an entangled state. For large $n$ and $d$, one obtains the Heisenberg scaling
\begin{equation}
{\left( \partial_{d}\langle A \rangle_{\ket{\psi_{n}(d)}}\right)^{2} \over \text{Var}_{\ket{\psi_{n}(d)}}A} \sim {256 n^{2}\over d^{8}}
\end{equation}
which exhibits Heisenberg scaling with the excitation number $n$, but also a polynomial reduction with respect to $d$ compared to (\ref{eqn:qfqfqf}). 

\subsection{Loss of Heisenberg scaling under noisy quantum channels}

Damping of the excitations of a massive quantum oscillator can be described by the quantum channel $\mathcal{N}_{\gamma}$ acting on a massive oscillator state $\rho$ by
\begin{equation}
\mathcal{N}_{\gamma}(\rho):=\sum_{j=0}^{\infty}{\gamma^{j}\over j!}(1-\gamma)^{a^{\dagger}a\over 2}a^{j}\rho a^{\dagger j}(1-\gamma)^{a^{\dagger}a\over 2}
\end{equation}
where $a\ket{\psi_{n}(d)}=\sqrt{n}\ket{\psi_{n-1}(d)}$, $a^{\dagger}\ket{\psi_{n}(d)}=\sqrt{n+1}\ket{\psi_{n+1}(d)}$ defines the action of the annihilation and creation operators \cite{ueda}. The loss parameter $\gamma$ interpolates between $\gamma=0$ (corresponding to the identity channel), and $\gamma=1$  (which destroys all components $\rho_{n,m}$ with $n,m>0$). 

Using the data processing inequality for the QFI \cite{PhysRevA.90.014101} or the fact that $\mathcal{N}_{\gamma}(\rho)$ is the solution of a Markovian open quantum system dynamics \cite{7442587,PhysRevA.82.042103}, one concludes that the QFI for $\mathcal{N}_{\gamma}(\rho)$ is bounded above by the QFI of $\rho$. Computation of the QFI of 
\begin{equation}
    \rho_{d}:= \mathcal{N}_{\gamma}(\psi_{n}(d)) = \sum_{j=0}^{\infty}{n\choose j}\gamma^{j}(1-\gamma)^{n-j}\psi_{n-j}(d)
    \label{eqn:noisystate}
\end{equation} is complicated by the fact that, according to Lemma \ref{lem:ooo}, the rank of $\partial_{d}\rho_{d}$ is greater than the rank of $\rho_{d}$. This precludes the use of a commonly encountered formula for the QFI of mixed probe states because the support of $\rho_{d}$ is a strict subspace of the support of the SLD operator \cite{paris,PhysRevA.95.052320}. However, the fact that the state $\rho_{d}$ depends on $d$ only through its eigenvectors motivates decomposing $\partial_{d}\rho_{d}$ into blocks according to the orthogonal projections $P_{d}$ and $Q_{d}$ onto $\text{Im}\rho_{d}$ and $\text{Im}\partial_{d}\rho_{d} \ominus \text{Im}\rho_{d}$, respectively. This procedure leads to an SLD operator
\begin{align}
    L_{d}&=2(R_{\rho_{d}} + L_{\rho_{d}})^{-1}P_{d}\partial_{d}\rho_{d} P_{d} \nonumber \\
    &+ 2R_{\rho_{d}}^{-1} Q_{d}\partial_{d}\rho_{d} P_{d} \nonumber \\
    &+ 2L_{\rho_{d}}^{-1} P_{d}\partial_{d}\rho_{d} Q_{d}
    \label{eqn:noisesld}
\end{align}
where $R_{O}$ ($L_{O}$) superoperators are the right (left) multiplication by operator $O$.

From the SLD (\ref{eqn:noisesld}), the QFI of probe state (\ref{eqn:noisystate}) is calculated for $n\ge 2$ as
\begin{align}
    \text{QFI}(d)&={n^{2}+n+1\over 2d^{2}} -{\gamma n^{2}\over d^{2}}+o(\gamma).
    \label{eqn:noisyqfi}
\end{align}
This equation shows that the Heisenberg scaling is destroyed linearly with increasing $\gamma$.

Because excitation damping reduces the excited state support of the probe state $\rho_{d}$, it is not surprising that Heisenberg scaling is destroyed with increasing $\gamma$. However, there are noise channels that destroy Heisenberg scaling of the QFI in more interesting ways. Especially in the context of motional heating of systems of electrically trapped ions that we discuss in Section \ref{sec:ti} as a physical system which can be used to investigate length scale estimation problems, it is relevant to consider thermalizing channels, namely, a continuous one-parameter family of length scale-independent quantum channels $\Phi_{t}$ that map the probe state $\rho_{d}$ to the thermal state. Mathematically,
\begin{equation}
\lim_{t\rightarrow \infty}\Vert \Phi_{t}(\rho_{d}) - \sum_{n=0}^{\infty}p_{n}\psi_{n}(d) \Vert_{1} = 0
\end{equation}
where $p_{n}:=\xi^{n}(1-\xi)$ is the geometric distribution on $\mathbb{Z}_{\ge 0}$ with parameter $\xi$ (physically, $-\ln \xi \in (0,\infty)$ is proportional to the reciprocal of the system temperature). Continuity of $\Phi_{t}$ (following from Stinespring's dilation theorem \cite{br}) and continuity of the probe state and its derivative with respect to $d$ imply that the states $\Phi_{t}(\rho_{d})$ and the operator $\Phi_{t}(\partial_{d}\rho_{d})$ are continuous with respect to $d$. This fact allows one to leverage the continuity of the QFI \cite{PhysRevA.100.032317} to show that the QFI of $\Phi_{t}(\rho_{d})$ converges to the QFI of the thermal state. However, the thermal state has an alternative expression
\begin{equation}
    \sum_{n=0}^{\infty}p_{n}\psi_{n}(d) = {1-\xi \over \pi \xi}\int_{\mathbb{C}} d^{2}\alpha \, e^{-{(1-\xi)\vert \alpha \vert^{2}\over \xi}} \psi_{\alpha}(d)
\end{equation}
with $d^{2}\alpha = d\text{Re}\alpha d\text{Im}\alpha$. Convexity of the QFI \cite{PhysRevA.63.042304} and the result in (\ref{eqn:cohqfi}) above then implies that the QFI of the thermal state is upper bounded by ${1\over d^{2}}O\left( {\xi \over 1-\xi} \right)$, i.e., by a linear function of the expected excitation number. This shows that Heisenberg scaling of $\rho_{d}$ is destroyed by $\Phi_{t}$ in a time that depends on the convergence rate of $\Phi_{t}(\rho_{d})$. In physical terms, the geometric distribution of excitations in a thermal state suppress the Heisenberg scaling of the QFI of its eigenvectors.

\subsection{Relation to quantum estimation in the scale model}
Lemma \ref{lem:ooo} can be extended by linearity to the state space of the massive oscillator with parameter $d$ by integrating to yield a Schrodinger-von Neumann equation with a logarithmic ``time''
\begin{eqnarray}
    {\partial \rho_{d} \over \partial \ln d}= -{i\over 4}[ia^{2}-ia^{\dagger 2},\rho_{d}].
\end{eqnarray}
Upon integrating, one finds that unlike the shift model in which the quantum state is parametrized by the additive abelian group $(\mathbb{R},+)$ according to $\rho_{x+y}=e^{-ixA}\rho_{y}e^{-ixA}$ for self-adjoint operator $A$ \cite{holevo}, the length scale estimation problem is in a scale model in which the parametrization takes the form
\begin{eqnarray}
    \rho_{dd'} = e^{-i\ln(d)A}\rho_{d'}e^{i\ln(d)A}
    \label{eqn:scalemodel}
\end{eqnarray}
with $d,d'$ in the multiplicative abelian group $(\mathbb{R}_{>0},\cdot)$ and self-adjoint operator $A={1\over 4}(ia^{2}-ia^{\dagger 2})$. This observation furnishes another proof of Theorem \ref{thm:hhh}, namely the QFI with respect to the parameter $\ln d$ is $4\text{Var}_{\ket{\psi_{n}(d)}}A$, which is in fact independent of $d$, and then the change of variables $\ln d \mapsto d$ gives $\text{QFI}(d)={4\text{Var}_{\ket{\psi_{n}(1)}}A\over d^{2}}$, which simplifies to (\ref{eqn:qfqfqf}). From this view, the $d$ dependence of the QFI arises from the reparametrization of the state manifold. The scale model (\ref{eqn:scalemodel}) allows one to consider the problem of estimation of $d$ by first considering states as parametrized by $\ln(d)$, then employing the covariant framework of Ref.\cite{PhysRevA.73.062103}, then finally transforming back to the coordinate $d$. For example, the scale invariant cost function \begin{eqnarray}
    c\left(\ln\tilde{d},\ln d\right)=-\int_{\epsilon}^{\infty}d\mu \, { \cos (\mu \vert \ln \left( {\tilde{d}\over d} \right) \vert) \over \mu}
    \label{eqn:covcost}
\end{eqnarray} (with $\epsilon>0$) is a function of $\vert \ln\left( {\tilde{d}\over d}\right)\vert$ which is both appropriate for scale estimation problems \cite{rubio2022quantum,PhysRevLett.127.190402} and allows to prove optimality of the covariant measurement.  

The optimal covariant measurement of $\ln(d)$ depends on the probe state and it is not clear how to implement the projections involved using standard readouts in quantum optics. In the next section, we focus on a more practical method for achieving the Heisenberg scaling predicted by Theorem \ref{thm:hhh} by analyzing asymptotically unbiased and efficient estimators of $d$ for arbitrary $d$. 

\section{Estimation theory\label{sec:esttheor}}

Because the length scale estimation problem involves a single parameter, the quantum Cram\'{e}r-Rao bound is saturated by a measurement consisting of eigenvectors of the SLD operator, in the sense that the classical Fisher information of this measurement is equal to the QFI. However, this fact is not particularly useful for designing a measurement to estimate $d$ because its value must be known in advance. Fortunately, for the probe states analyzed in the present work, the use of a function of $q$ as an estimator for these probe states is justified by the following theorem.
\begin{theorem}
    The classical Fisher information of the position measurement $\lbrace \ket{q=x}\bra{q=x},dx\rbrace$ for a pure probe state $\ket{\psi(d)}$ with real-valued position wavefunction is equal to $\mathrm{QFI}(d)$ for all $d>0$.
    \label{thm:ttt}
\end{theorem}

\begin{proof} We appeal to a result of \cite{PhysRevLett.72.3439} as stated in Eq. (28) of \cite{bng}, which provides a sufficient and necessary condition for the QFI to be equal to the classical Fisher information for estimation of a single parameter with a real-valued quantum measurement $\lbrace m(x),dx\rbrace$. Specifically, the condition is
\begin{eqnarray}
    m(x)^{1/2}L_{d}\rho_{d}^{1/2} &=& r(x)m(x)^{1/2}\rho_{d}^{1/2},
    \label{eqn:bneq}
\end{eqnarray}
almost surely with $r(x)\in\mathbb{R}$.

We take $m(x):= \ket{q=x}\bra{q=x}$. Because $m(x)$ is a limit of rank-one projections, $m(x) = m(x)^{1/2}$. Replacing matrix multiplication in the left-hand side with integration in (\ref{eqn:bneq}), it suffices to determine $r(x) \in \mathbb{R}$ such that: \begin{eqnarray}
    \int L_{d}(x,z)\rho_{d}^{1/2}(z,y) dz&=& r(x)\rho^{1/2}_{d}(x,y)
\end{eqnarray}
where we write $A(x,y):= \langle q=x\vert A\vert q=y\rangle$ for the position kernel of an operator $A$. Since $\rho_{d}$ is assumed pure, $\rho_{d}^{1/2} = \rho_{d}$ and $L_{d}=2\partial_{d}\rho_{d}$. Hence we may further simplify to:

\begin{eqnarray}
\int 2\partial_{d} \rho_d(x,z)\rho_{d}(z,y) dz &=& r(x)\rho_{d}(x,y).
\end{eqnarray}

With $\rho_{d}(x,y) := \psi_{d}(x) \psi_{d}(y)$, this becomes
\begin{eqnarray}
\int 2\partial_{d}[\psi_{d}(x)\psi_{d}(z)]\psi_{d}(z) dz &=& r(x)\psi_{d}(x).
\end{eqnarray}
Solving for $r(x)$ shows that it is real-valued where it is defined.\end{proof}

Rotating both the phase space measurement direction and the probe state, one further concludes that for pure states parametrized by the oscillator length, the classical Fisher information about the parameter $d$ for a measurement of displacement along any phase space direction is equal to $\text{QFI}(d)$ if the wavefunction is real in that phase space quadrature. For example, the classical Fisher information (for parameter $d$) of a homodyne measurement is equal to the $\text{QFI}(d)$ if and only if $\alpha$ is real-valued. 

Although Theorem \ref{thm:ttt} establishes a $q$-measurement as optimal for estimation of $d$, it does not provide an estimator. For the probe state $\ket{\psi_{n}(d)}$, we consider both the method of moments (MOM) estimator and the maximum likelihood estimator (MLE) based on independent, identical $q$ measurements; the MOM and MLE are two of the most common estimators in general \cite{lehmanntheory}. For $M$ measurements, the MOM estimator of ${1\over d}$ is given by $X_{M}:=\frac{2\sum{q_i^2}}{M(2n+1)}$, which is derived by setting the second population moment $E[q^{2}]:=\langle \psi_{n}(d) \vert q^{2}\vert \psi_{n}(d)\rangle$ to the sample second moment. From the expressions $\textrm{E}[q^2]$ and $\textrm{Var}[q^2]$ derived in Theorem \ref{thm:hhh}, the MOM estimator is an unbiased estimator of $1/d$ with asymptotic variance ${2(n^{2}+n+1)\over d^{2}(2n+1)^{2}}$. Applying the delta method \cite{lehmanntheory} with the function $X_{M}\mapsto {1\over X_{M}}$ gives an asymptotically unbiased estimator of $d$ with asymptotic variance ${2(n^{2}+n+1)d^{2}\over (2n+1)^{2}}$, which does not attain the Cram\'{e}r-Rao lower bound, i.e., the reciprocal of the $q$-measurement Fisher information (equal to the reciprocal of the QFI as per Theorem \ref{thm:ttt}), for any excitation number $n>1$. For $n=1$, the asymptotic variance coincides with the reciprocal of the QFI and Appendix \ref{sec:app1} justifies this fact by noting that equivalence between the MOM and MLE for this specific excitation number. 

In contrast to the MOM estimator, the MLE is known to be asymptotically efficient subject to regularity conditions \cite{van2000asymptotic, shao2008mathematical}. Specifically, if $\hat{d}$ is the MLE for $M$ samples of $\rho_{d}$, then $\sqrt{M}(\hat{d} - d)$ converges in distribution to $N(0, 1/\textrm{QFI}(d))$, implying that the bias goes to 0 in order less than $1/\sqrt{M}$ and the variance of $\hat{d}$ is approximately  $\frac{1}{M \times \text{QFI}(d)}$, reducing to 0 in order $1/M$. For a probe state $\ket{\psi_{n}(d)}$, the MLE is defined as the value of $d$ which maximizes the log-likelihood $l_{M}(d)$ of the probability distribution $\prod_{i=1}^{M}\langle q=q_{i}\vert \psi_{n}(d)\rangle^{2}$, namely the function
\begin{equation}
    l_{M}(d) = \frac{M}{2}\log(d) - d \sum_{i=1}^M q_i^2 + 2\sum_{i=1}^M\log\vert H_n(q_i \sqrt{d})\vert,
\end{equation}
dropping constants that don't include $d$ (also, log may correspond to any base, but usually the natural logarithm is used). In simulations, we verified that a numerical MLE, with initial guess starting at the MOM, and the MOM estimator are asymptotically unbiased, and that the variance of the numerical MLE approaches the quantum Cram\'{e}r-Rao  bound, whereas the variance of the MOM does not scale with excitation number, $n$. 

Our investigation has been for the general excited state, $\ket{\psi_{n}(d)}$, and the probability density function (PDF) for $q$ is not in the exponential family of distributions for $n \geq 2$. However, the $n = 0$ case is trivially Gaussian and the $n = 1$ case can also be shown to be in the exponential family of distributions. We examine some further properties of the $n = 1$ excited state in Appendix \ref{sec:app1}, including a Bayesian posterior.

\section{Length scale estimation with multiple massive oscillators\label{sec:exc}}
In the shift model of estimation of a single-parameter SU(2) rotation acting on $(\mathbb{C}^{2})^{\otimes n}$, the fact that entangled states provide the optimal probes is a simple consequence of the fact that an equally-weighted superposition of the highest and lowest weight vectors for the rotation generator occurs only for states that are entangled across any bipartition of the $n$ registers.

Outside of the shift model, it becomes unclear whether entangled states provide an advantage because there is not necessarily a generator associated with the parameter. In this case, one strategy is to consider how entanglement-generating unitaries affect the parametrization of a probe state. In the case of estimation of $d$ with two-mode entangled Gaussian states, the Euler decomposition implies that a probe state has the form $U\ket{\psi(d)}_{A}\ket{\phi(d)}_{B}$ where $U$ acts as an element of $U(2)\cong O(4)\cap Sp(4,\mathbb{R})$ and $\ket{\psi(d)}$ and $\ket{\phi(d)}$ are single-oscillator Gaussian states. But the Cartan decomposition of $U(2)$ allows to simplify $U$ as a product of a unitary that is generated by $q_{A}p_{B}$, $p_{A}q_{B}$ and another unitary which applies local phase shifts to $\ket{\psi(d)}\ket{\phi(d)}$. Since $q_{A}p_{B}$, $p_{A}q_{B}$ are independent of $d$, entanglement-generating linear optical Gaussian transformations cannot increase the $\text{QFI}(d)$ of an optimal product Gaussian state.

On the other hand, for general non-Gaussian probe states it is instructive to analyze the QFI for  examples of entangled states of energy $n+m$ constructed from the states $\ket{\psi_{n}(d)}$ and $\ket{\psi_{m}(d)}$, and compare to the QFI of the product state $\ket{\psi_{n}(d)}\ket{\psi_{m}(d)}$. The states
\begin{align}
&{} {\mathbb{I}\pm\text{SWAP} \over \sqrt{2}}\ket{\psi_{n}(d)}\ket{\psi_{m}(d)} \label{eqn:sub1}\\
&{}    {\ket{\psi_{n}(d)}^{\otimes 2} \pm \ket{\psi_{m}(d)}^{\otimes 2} \over \sqrt{2}}
\label{eqn:sub2}
\end{align}
are entangled states of two massive oscillators of total energy $n+m$ and are motivated by the fact that they are maximally entangled in the subspace $\text{span}_{\mathbb{C}}\lbrace \ket{\psi_{n}(d)},\ket{\psi_{m}(d)}\rbrace$. Using Lemma \ref{lem:ooo}, it is straightforward to show that for $\vert n-m\vert >2$, the QFI of these probe states is the sum of the values in (\ref{eqn:qfqfqf}) for $\ket{\psi_{n}(d)}$ and $\ket{\psi_{m}(d)}$, and is therefore equal to the QFI of the product state $\ket{\psi_{n}(d)}\ket{\psi_{m}(d)}$. Note that the single oscillator state $\ket{\psi_{m+n}(d)}$ has QFI greater than the sum of the values in (\ref{eqn:qfqfqf}) for the probe states $\ket{\psi_{n}(d)}$ and $\ket{\psi_{m}(d)}$ when $m,n\neq 0$. However, comparison of the QFI of the states (\ref{eqn:sub2}) to the QFI of $\ket{\psi_{m+n}(d)}$ does not appear to be meaningful because $\ket{\psi_{m+n}(d)}$ describes a different number of modes and has  less energy than $\ket{\psi_{n}(d)}\ket{\psi_{m}(d)}$ (appending a second mode comes with its zero-point energy). Considering instead $\vert n-m\vert =2$ in the states in (\ref{eqn:sub2}), there is a entanglement advantage for length scale estimation. For example, with $m=n\pm 2$ one can consider the minus sign in (\ref{eqn:sub2}) or the plus sign in (\ref{eqn:sub1})  and obtain an $O(n^{2})$ additive increase in the QFI compared to $\ket{\psi_{n}(d)}\ket{\psi_{n\pm 2}(d)}$. To extending the form of these entangled states to states having amplitude centered on a high excitation number, one can consider a sequence of excitation numbers $n(\ell)$ with $\ell\in \mathbb{N}$. Then if $n(\ell)>2\ell$, the additive increase in the QFI for the sequence of entangled states \begin{equation}\ket{\phi^{(\ell)}(d)}:={1\over \sqrt{2\ell+1}}\sum_{j=-\ell}^{\ell}(-1)^{\delta_{j,0}}\ket{\psi_{n(\ell)+2j}(d)}^{\otimes 2}\label{eqn:entseq}\end{equation} over the QFI of $\ket{\psi_{n(\ell)}(d)}^{\otimes 2}$ has the form ${2\ell n(\ell)^{2}\over 2d^{2}(2\ell+1)} + o(n(\ell)^{2})$. It then follows from Theorem \ref{thm:hhh} that compared to the sequence of states $\ket{\psi_{n(\ell)}(d)}^{\otimes 2}$, bipartite entanglement asymptotically improves the sensitivity as much as appending an additional mode prepared in $\ket{\psi_{n(\ell)}(d)}$.

The fact that the excited states $\ket{\psi_{n}(d)}$ allow to achieve Heisenberg scaling and that entangled states provide a further advantage in achievable precision motivates the question of existence of a probe state of a system of $N$ massive oscillators that allows to achieve the global Heisenberg scaling $O\left( {N^{2}n^{2}\over d^{2}}\right)$ of the QFI, i.e., QFI scaling as the square of the total energy of the system. From additivity of the QFI under tensor products, one expects such probe states to be multipartite entangled analogously to Greenberger-Horne-Zeilinger (GHZ) states. As in the problem of multimode optical phase estimation with states of fixed total photon number \cite{PhysRevLett.111.070403}, we demand that the probe state be symmetric. However, we do not demand that the state be a superposition of states of equal total excitation number because the relevant superselection rule for a system of $N$ massive oscillators is simply the total mass $Nm$. In other words, the excitations can be considered as internal degrees of freedom of the system subject to, e.g., Rabi dynamics. A naive generalization of, e.g., (\ref{eqn:sub2}) with $m=n+2$, to form the states $\ket{n}^{\otimes N} + \ket{n+2}^{\otimes N}$ (normalization omitted) fails because, although a small change in $d$ locally perturbs the $\ket{n}$ and $\ket{n+2}$ states, the global state does not change with a rate $O(N^{2})$, which is the maximal scaling that one would expect occur for a macroscopic superposition state \cite{PhysRevA.89.012122}. A probe state with the desired behavior can be constructed by taking a superposition of two states that incur an amplitude change as $d$ changes. To this end, consider the state
\begin{equation}
{1\over \sqrt{2}}\left( \left( {\ket{n}+i\ket{n+2}\over \sqrt{2}}\right)^{\otimes N}+\left( {\ket{n}-i\ket{n+2}\over \sqrt{2}}\right)^{\otimes N} \right)
\label{eqn:ghzghz}
\end{equation}
which is a superposition of two states of $N$ massive oscillators with equal expected excitation number $N(n+1)$. Note that the corresponding wavefunction on $\mathbb{R}^{N}$ is real, which allows one to leverage Theorem \ref{thm:ttt} in finding an optimal readout. A global change in $\text{ln}d$ is generated by the operator $\sum_{j=1}^{N}ia^{2}_{j} - ia^{\dagger 2}_{j}$ which, when projected to the two-dimensional subspace spanned by $\ket{n_{\pm}}:={\ket{n}\pm \ket{n+2}\over \sqrt{2}}$, takes the form $\sqrt{(n+2)(n+1)}\sum_{j=1}^{N}Z_{j}$, with $Z_{j}$ the Pauli $Z$ operator on the $j$-th mode (and identity elsewhere). By taking the variance of this operator, then transforming back to the coordinate $d$ from $\ln d$, it is straightforward to verify that the probe state (\ref{eqn:ghzghz}) has QFI scaling as $\text{QFI}(d)=O\left( {N^{2}n^{2}\over d^{2}}\right)$.

\section{Experimental directions and challenges\label{sec:ti}}

Our result in Theorem \ref{thm:hhh} and subsequent discussion of measurement readouts indicates that an experiment combining a protocol for preparation of Fock states $\ket{\psi_{n}(d)}$ of a system of massive oscillators with a protocol for finite-energy Bell measurements (\ref{eqn:bell}) is sufficient for realizing Heisenberg scaling in precision for an estimate of $d$. Trapped ions provide a system in which substantial experimental progress has been made toward preparing non-classical states of massive oscillators for use in quantum estimation protocols. Specifically, in a one-dimensional system of trapped ions cooled into the phonon-mediated interaction regime, the magnitude of the spin-dependent force is proportional to the length scale characterizing the phonon oscillator because the spin operators couple to the canonical operator $q$ representing displacement of the phonon field  \cite{PhysRevLett.122.030501}. This spin-dependent force is the mechanism for implementation of multiqubit gates when using the trapped ions for quantum computation \cite{PhysRevLett.82.1835,PhysRevLett.74.4091}. Estimates of $d$ obtained from a quantum estimation protocol would be used in pulse sequences that implement multiqubit gates on the spin degrees of freedom. Therefore, lowering the variance of the estimates toward the quantum Cram\'{e}r-Rao bounds determined by the $\text{QFI}(d)$ in Sections \ref{sec:bench}, \ref{sec:hs}, \ref{sec:exc} is expected to increase the multiqubit gate fidelities for quantum computation with trapped ions.

Recent experiments in trapped ion systems have generated $n=O(1)$ Fock states \cite{Um2016} for the purpose of trap frequency and displacement estimation \cite{Wolf2019}. In a trapped ion system, phase space characterization of non-classical states has been achieved by tomographic reconstruction of the Wigner quasiprobability function \cite{PhysRevLett.125.043602} and characterization of Fock space support has been achieved with number-resolving phonon measurement \cite{PhysRevLett.131.223603}. Phonon squeezing operations implemented by parametric  amplification have been proposed for improving the multi-qubit gate fidelity and spin-squeezing properties of ensembles of trapped ions \cite{PhysRevLett.122.030501,PRXQuantum.4.030311,PhysRevLett.132.163601,PhysRevLett.129.063603,PhysRevA.93.013415}. This proposal has been implemented with a single-ion mechanical oscillator, achieving displacement estimation with precision exceeding the standard quantum limit \cite{wineland}, and progress has been made toward demonstrating analogous results in larger trapped ion systems \cite{PhysRevA.107.032425}. Along these lines of research, demonstrating the entangled probe states in Section \ref{sec:exc} and their associated readouts provide challenges for scaling continuous-variable quantum information processing based on massive quantum oscillators to larger system sizes and composite systems. In trapped ion systems, progress in this direction has been made by utilizing optical Bell measurements to generate entanglement between ion systems \cite{Hucul2015} and by implementing mixed-species entangling gates \cite{PhysRevLett.125.080504}. Proposals also exist for entangled state and phonon Fock state generation in optomechanical systems \cite{Clarke_2020,PhysRevA.99.043836}.

In quantum computation and communication, quantum error correcting codes provide fault-tolerance necessary to overcome decoherence during the computation. In the context of continuous-variable systems,
Gottesman-Kitaev-Preskill (GKP) codes provide a versatile quantum error correction scheme \cite{PhysRevA.64.012310}, and a method for preparing phonon GKP code states has been demonstrated in a single-ion system \cite{PhysRevLett.133.050602}. Recently developed sequential squeezing protocols \cite{PRXQuantum.5.020314} could also be useful for generating such code states. To get accurate readouts for displacement errors, and therefore
a higher recovery rate, motional dephasing should be characterized during the pulse sequences of interest. This motional dephasing can arise from electronic noise in the radio-frequency trap, which can be related to a fluctuating oscillator length scale parameter. Since
the low phonon occupation GKP code states are not expected to be useful for high-precision length scale estimation, one could consider dedicating a trapped ion subsystem to relatively high occupation states considered in the previous sections. This hypothetical system could support fault-tolerant computation (by error correction) and noise characterization (by length scale estimation) in a trapped ion framework.

\section{Discussion}

Our results on advantages in quantum estimation of the length scale parameter with non-classical and entangled massive oscillators motivate further development of quantum-enhanced sensing protocols with continuous variables quantum systems. The length scale, or, more generally, the complex structure, initially appears to be a non-dynamical parameter for massive quantum continuous-variable systems. Indeed, some intuition from the standard Gaussian estimation paradigm such as squeezing enhancement of sensitivity is not applicable because the generator does not couple to the length scale. However, the derivatives of the wavefunctions of the excitations shows that $\ln d$ can be considered as a dynamical parameter, allowing one to leverage standard techniques from quantum estimation theory, at least to derive bounds on the ultimate attainable precision. In terms of experimental requirements, the structure of the states (\ref{eqn:sub1}), (\ref{eqn:sub2}), (\ref{eqn:ghzghz}) indicate that coherent control in a two-dimensional subspace of even or odd parity excitations is sufficient for realizing our predictions.

In the present work, we did not consider estimation of multiple length scale parameters that may be present in a heterogeneous system of massive quantum oscillators. Our construction of GHZ-like states that provide global Heisenberg scaling of QFI suggests that estimation of a function of the different length scales can be analyzed along the lines of discrete variable distributed sensing \cite{PhysRevLett.120.080501} with adaptive maximum likelihood estimation providing a possible route for achieving optimal sensitivity \cite{fujiwara}.

Dynamics of interacting massive oscillators are important targets for digital quantum simulation algorithms. The present work indicates several targets for wavefunction preparation algorithms that could be defined in the continuous-variable domain and simulated digitally.  We note that when the normalization factor of a discretized Gaussian wavefunction of a massive quantum oscillator can be efficiently computed, an approximation of the wavefunction can be prepared on a digital quantum computer by an algorithm with depth $O(\log{1\over \epsilon})$ where $\epsilon$ is the desired precision \cite{grovrud}. Further, the wavefunction of the same state (as a sequence in $\ell^{2}(\mathbb{C})$), but corresponding to a different mass, can be prepared using the linear resampling algorithm in Ref.\cite{kitaev}. Such protocols can be used to prepare initial states for simulations of quantum dynamics that generate the states analyzed in this work.

\begin{acknowledgements} 
This work was supported in part by the Laboratory Directed Research and Development program of Los Alamos National Laboratory, via an Information Science and Technology Institute (ISTI) rapid response award. T.V. acknowledges additional support from the National Quantum Information Science Research Centers and the Quantum Science Center, and the LANL ASC Beyond Moore’s Law project.  

\end{acknowledgements}
%\clearpage

%\onecolumngrid
\bibliography{refs2.bib}

\begin{thebibliography}{70}%
\makeatletter
\providecommand \@ifxundefined [1]{%
 \@ifx{#1\undefined}
}%
\providecommand \@ifnum [1]{%
 \ifnum #1\expandafter \@firstoftwo
 \else \expandafter \@secondoftwo
 \fi
}%
\providecommand \@ifx [1]{%
 \ifx #1\expandafter \@firstoftwo
 \else \expandafter \@secondoftwo
 \fi
}%
\providecommand \natexlab [1]{#1}%
\providecommand \enquote  [1]{``#1''}%
\providecommand \bibnamefont  [1]{#1}%
\providecommand \bibfnamefont [1]{#1}%
\providecommand \citenamefont [1]{#1}%
\providecommand \href@noop [0]{\@secondoftwo}%
\providecommand \href [0]{\begingroup \@sanitize@url \@href}%
\providecommand \@href[1]{\@@startlink{#1}\@@href}%
\providecommand \@@href[1]{\endgroup#1\@@endlink}%
\providecommand \@sanitize@url [0]{\catcode `\\12\catcode `\$12\catcode
  `\&12\catcode `\#12\catcode `\^12\catcode `\_12\catcode `\%12\relax}%
\providecommand \@@startlink[1]{}%
\providecommand \@@endlink[0]{}%
\providecommand \url  [0]{\begingroup\@sanitize@url \@url }%
\providecommand \@url [1]{\endgroup\@href {#1}{\urlprefix }}%
\providecommand \urlprefix  [0]{URL }%
\providecommand \Eprint [0]{\href }%
\providecommand \doibase [0]{https://doi.org/}%
\providecommand \selectlanguage [0]{\@gobble}%
\providecommand \bibinfo  [0]{\@secondoftwo}%
\providecommand \bibfield  [0]{\@secondoftwo}%
\providecommand \translation [1]{[#1]}%
\providecommand \BibitemOpen [0]{}%
\providecommand \bibitemStop [0]{}%
\providecommand \bibitemNoStop [0]{.\EOS\space}%
\providecommand \EOS [0]{\spacefactor3000\relax}%
\providecommand \BibitemShut  [1]{\csname bibitem#1\endcsname}%
\let\auto@bib@innerbib\@empty
%</preamble>
\bibitem [{\citenamefont {Moshinsky}(1962)}]{MOSHINSKY1962384}%
  \BibitemOpen
  \bibfield  {author} {\bibinfo {author} {\bibfnamefont {M.}~\bibnamefont
  {Moshinsky}},\ }\bibfield  {title} {\bibinfo {title} {The harmonic oscillator
  and supermultiplet theory: (i) the single shell picture},\ }\href
  {https://doi.org/https://doi.org/10.1016/0029-5582(62)90758-7} {\bibfield
  {journal} {\bibinfo  {journal} {Nuclear Physics}\ }\textbf {\bibinfo {volume}
  {31}},\ \bibinfo {pages} {384} (\bibinfo {year} {1962})}\BibitemShut
  {NoStop}%
\bibitem [{\citenamefont {Élliott}\ \emph {et~al.}(2005)\citenamefont
  {Élliott}, \citenamefont {Seiringer}, \citenamefont {Solovej},\ and\
  \citenamefont {Yngvason}}]{lieb}%
  \BibitemOpen
  \bibfield  {author} {\bibinfo {author} {\bibfnamefont {L.}~\bibnamefont
  {Élliott}}, \bibinfo {author} {\bibfnamefont {R.}~\bibnamefont {Seiringer}},
  \bibinfo {author} {\bibfnamefont {J.}~\bibnamefont {Solovej}},\ and\ \bibinfo
  {author} {\bibfnamefont {J.}~\bibnamefont {Yngvason}},\ }\href
  {https://doi.org/10.1007/b137508} {\emph {\bibinfo {title} {The mathematics
  of the {Bose} gas and its condensation}}}\ (\bibinfo  {publisher}
  {Birkh\"{a}user Verlag},\ \bibinfo {year} {2005})\BibitemShut {NoStop}%
\bibitem [{\citenamefont {Wineland}\ and\ \citenamefont
  {Leibfried}(2011)}]{winewine}%
  \BibitemOpen
  \bibfield  {author} {\bibinfo {author} {\bibfnamefont {D.}~\bibnamefont
  {Wineland}}\ and\ \bibinfo {author} {\bibfnamefont {D.}~\bibnamefont
  {Leibfried}},\ }\bibfield  {title} {\bibinfo {title} {Quantum information
  processing and metrology with trapped ions},\ }\href
  {https://doi.org/https://doi.org/10.1002/lapl.201010125} {\bibfield
  {journal} {\bibinfo  {journal} {Laser Physics Letters}\ }\textbf {\bibinfo
  {volume} {8}},\ \bibinfo {pages} {175} (\bibinfo {year} {2011})}\BibitemShut
  {NoStop}%
\bibitem [{\citenamefont {Gr{\"o}blacher}\ \emph {et~al.}(2009)\citenamefont
  {Gr{\"o}blacher}, \citenamefont {Hammerer}, \citenamefont {Vanner},\ and\
  \citenamefont {Aspelmeyer}}]{Groblacher2009}%
  \BibitemOpen
  \bibfield  {author} {\bibinfo {author} {\bibfnamefont {S.}~\bibnamefont
  {Gr{\"o}blacher}}, \bibinfo {author} {\bibfnamefont {K.}~\bibnamefont
  {Hammerer}}, \bibinfo {author} {\bibfnamefont {M.~R.}\ \bibnamefont
  {Vanner}},\ and\ \bibinfo {author} {\bibfnamefont {M.}~\bibnamefont
  {Aspelmeyer}},\ }\bibfield  {title} {\bibinfo {title} {Observation of strong
  coupling between a micromechanical resonator and an optical cavity field},\
  }\href {https://doi.org/10.1038/nature08171} {\bibfield  {journal} {\bibinfo
  {journal} {Nature}\ }\textbf {\bibinfo {volume} {460}},\ \bibinfo {pages}
  {724} (\bibinfo {year} {2009})}\BibitemShut {NoStop}%
\bibitem [{\citenamefont {Samanta}\ \emph {et~al.}(2023)\citenamefont
  {Samanta}, \citenamefont {De~Bonis}, \citenamefont {M{\o}ller}, \citenamefont
  {Tormo-Queralt}, \citenamefont {Yang}, \citenamefont {Urgell}, \citenamefont
  {Stamenic}, \citenamefont {Thibeault}, \citenamefont {Jin}, \citenamefont
  {Czaplewski}, \citenamefont {Pistolesi},\ and\ \citenamefont
  {Bachtold}}]{Samanta2023}%
  \BibitemOpen
  \bibfield  {author} {\bibinfo {author} {\bibfnamefont {C.}~\bibnamefont
  {Samanta}}, \bibinfo {author} {\bibfnamefont {S.~L.}\ \bibnamefont
  {De~Bonis}}, \bibinfo {author} {\bibfnamefont {C.~B.}\ \bibnamefont
  {M{\o}ller}}, \bibinfo {author} {\bibfnamefont {R.}~\bibnamefont
  {Tormo-Queralt}}, \bibinfo {author} {\bibfnamefont {W.}~\bibnamefont {Yang}},
  \bibinfo {author} {\bibfnamefont {C.}~\bibnamefont {Urgell}}, \bibinfo
  {author} {\bibfnamefont {B.}~\bibnamefont {Stamenic}}, \bibinfo {author}
  {\bibfnamefont {B.}~\bibnamefont {Thibeault}}, \bibinfo {author}
  {\bibfnamefont {Y.}~\bibnamefont {Jin}}, \bibinfo {author} {\bibfnamefont
  {D.~A.}\ \bibnamefont {Czaplewski}}, \bibinfo {author} {\bibfnamefont
  {F.}~\bibnamefont {Pistolesi}},\ and\ \bibinfo {author} {\bibfnamefont
  {A.}~\bibnamefont {Bachtold}},\ }\bibfield  {title} {\bibinfo {title}
  {Nonlinear nanomechanical resonators approaching the quantum ground state},\
  }\href {https://doi.org/10.1038/s41567-023-02065-9} {\bibfield  {journal}
  {\bibinfo  {journal} {Nature Physics}\ }\textbf {\bibinfo {volume} {19}},\
  \bibinfo {pages} {1340} (\bibinfo {year} {2023})}\BibitemShut {NoStop}%
\bibitem [{\citenamefont {Rips}\ \emph {et~al.}(2014)\citenamefont {Rips},
  \citenamefont {Wilson-Rae},\ and\ \citenamefont
  {Hartmann}}]{PhysRevA.89.013854}%
  \BibitemOpen
  \bibfield  {author} {\bibinfo {author} {\bibfnamefont {S.}~\bibnamefont
  {Rips}}, \bibinfo {author} {\bibfnamefont {I.}~\bibnamefont {Wilson-Rae}},\
  and\ \bibinfo {author} {\bibfnamefont {M.~J.}\ \bibnamefont {Hartmann}},\
  }\bibfield  {title} {\bibinfo {title} {Nonlinear nanomechanical resonators
  for quantum optoelectromechanics},\ }\href
  {https://doi.org/10.1103/PhysRevA.89.013854} {\bibfield  {journal} {\bibinfo
  {journal} {Phys. Rev. A}\ }\textbf {\bibinfo {volume} {89}},\ \bibinfo
  {pages} {013854} (\bibinfo {year} {2014})}\BibitemShut {NoStop}%
\bibitem [{\citenamefont {Rips}\ \emph {et~al.}(2012)\citenamefont {Rips},
  \citenamefont {Kiffner}, \citenamefont {Wilson-Rae},\ and\ \citenamefont
  {Hartmann}}]{Rips_2012}%
  \BibitemOpen
  \bibfield  {author} {\bibinfo {author} {\bibfnamefont {S.}~\bibnamefont
  {Rips}}, \bibinfo {author} {\bibfnamefont {M.}~\bibnamefont {Kiffner}},
  \bibinfo {author} {\bibfnamefont {I.}~\bibnamefont {Wilson-Rae}},\ and\
  \bibinfo {author} {\bibfnamefont {M.~J.}\ \bibnamefont {Hartmann}},\
  }\bibfield  {title} {\bibinfo {title} {Steady-state negative wigner functions
  of nonlinear nanomechanical oscillators},\ }\href
  {https://doi.org/10.1088/1367-2630/14/2/023042} {\bibfield  {journal}
  {\bibinfo  {journal} {New Journal of Physics}\ }\textbf {\bibinfo {volume}
  {14}},\ \bibinfo {pages} {023042} (\bibinfo {year} {2012})}\BibitemShut
  {NoStop}%
\bibitem [{\citenamefont {Banerjee}\ \emph {et~al.}(2023)\citenamefont
  {Banerjee}, \citenamefont {Kalita},\ and\ \citenamefont
  {Sarma}}]{Banerjee:23}%
  \BibitemOpen
  \bibfield  {author} {\bibinfo {author} {\bibfnamefont {P.}~\bibnamefont
  {Banerjee}}, \bibinfo {author} {\bibfnamefont {S.}~\bibnamefont {Kalita}},\
  and\ \bibinfo {author} {\bibfnamefont {A.~K.}\ \bibnamefont {Sarma}},\
  }\bibfield  {title} {\bibinfo {title} {{Robust mechanical squeezing beyond 3
  dB in a quadratically coupled optomechanical system}},\ }\href
  {https://doi.org/10.1364/JOSAB.483944} {\bibfield  {journal} {\bibinfo
  {journal} {J. Opt. Soc. Am. B}\ }\textbf {\bibinfo {volume} {40}},\ \bibinfo
  {pages} {1398} (\bibinfo {year} {2023})}\BibitemShut {NoStop}%
\bibitem [{\citenamefont {Choi}\ \emph {et~al.}(2024)\citenamefont {Choi},
  \citenamefont {Pluchar}, \citenamefont {He}, \citenamefont {Guha},\ and\
  \citenamefont {Wilson}}]{choi2024quantumlimitedimagingnanomechanical}%
  \BibitemOpen
  \bibfield  {author} {\bibinfo {author} {\bibfnamefont {M.}~\bibnamefont
  {Choi}}, \bibinfo {author} {\bibfnamefont {C.}~\bibnamefont {Pluchar}},
  \bibinfo {author} {\bibfnamefont {W.}~\bibnamefont {He}}, \bibinfo {author}
  {\bibfnamefont {S.}~\bibnamefont {Guha}},\ and\ \bibinfo {author}
  {\bibfnamefont {D.}~\bibnamefont {Wilson}},\ }\href
  {https://arxiv.org/abs/2411.04980} {\bibinfo {title} {Quantum limited imaging
  of a nanomechanical resonator with a spatial mode sorter}} (\bibinfo {year}
  {2024}),\ \Eprint {https://arxiv.org/abs/2411.04980} {arXiv:2411.04980
  [quant-ph]} \BibitemShut {NoStop}%
\bibitem [{\citenamefont {Tang}\ \emph {et~al.}(2022)\citenamefont {Tang},
  \citenamefont {Cai}, \citenamefont {Cheng}, \citenamefont {Xu}, \citenamefont
  {Peng}, \citenamefont {Chen}, \citenamefont {Wang}, \citenamefont {Xia},
  \citenamefont {Wang}, \citenamefont {Song}, \citenamefont {Zhou},\ and\
  \citenamefont {Deng}}]{TANG2022127966}%
  \BibitemOpen
  \bibfield  {author} {\bibinfo {author} {\bibfnamefont {J.-D.}\ \bibnamefont
  {Tang}}, \bibinfo {author} {\bibfnamefont {Q.-Z.}\ \bibnamefont {Cai}},
  \bibinfo {author} {\bibfnamefont {Z.-D.}\ \bibnamefont {Cheng}}, \bibinfo
  {author} {\bibfnamefont {N.}~\bibnamefont {Xu}}, \bibinfo {author}
  {\bibfnamefont {G.-Y.}\ \bibnamefont {Peng}}, \bibinfo {author}
  {\bibfnamefont {P.-Q.}\ \bibnamefont {Chen}}, \bibinfo {author}
  {\bibfnamefont {D.-G.}\ \bibnamefont {Wang}}, \bibinfo {author}
  {\bibfnamefont {Z.-W.}\ \bibnamefont {Xia}}, \bibinfo {author} {\bibfnamefont
  {Y.}~\bibnamefont {Wang}}, \bibinfo {author} {\bibfnamefont {H.-Z.}\
  \bibnamefont {Song}}, \bibinfo {author} {\bibfnamefont {Q.}~\bibnamefont
  {Zhou}},\ and\ \bibinfo {author} {\bibfnamefont {G.-W.}\ \bibnamefont
  {Deng}},\ }\bibfield  {title} {\bibinfo {title} {A perspective on quantum
  entanglement in optomechanical systems},\ }\href
  {https://doi.org/https://doi.org/10.1016/j.physleta.2022.127966} {\bibfield
  {journal} {\bibinfo  {journal} {Physics Letters A}\ }\textbf {\bibinfo
  {volume} {429}},\ \bibinfo {pages} {127966} (\bibinfo {year}
  {2022})}\BibitemShut {NoStop}%
\bibitem [{\citenamefont {Sarma}\ \emph {et~al.}(2021)\citenamefont {Sarma},
  \citenamefont {Chakraborty},\ and\ \citenamefont
  {Kalita}}]{10.1116/5.0022349}%
  \BibitemOpen
  \bibfield  {author} {\bibinfo {author} {\bibfnamefont {A.~K.}\ \bibnamefont
  {Sarma}}, \bibinfo {author} {\bibfnamefont {S.}~\bibnamefont {Chakraborty}},\
  and\ \bibinfo {author} {\bibfnamefont {S.}~\bibnamefont {Kalita}},\
  }\bibfield  {title} {\bibinfo {title} {Continuous variable quantum
  entanglement in optomechanical systems: A short review},\ }\href
  {https://doi.org/10.1116/5.0022349} {\bibfield  {journal} {\bibinfo
  {journal} {AVS Quantum Science}\ }\textbf {\bibinfo {volume} {3}},\ \bibinfo
  {pages} {015901} (\bibinfo {year} {2021})},\ \Eprint
  {https://arxiv.org/abs/https://pubs.aip.org/avs/aqs/article-pdf/doi/10.1116/5.0022349/14572089/015901\_1\_online.pdf}
  {https://pubs.aip.org/avs/aqs/article-pdf/doi/10.1116/5.0022349/14572089/015901\_1\_online.pdf}
  \BibitemShut {NoStop}%
\bibitem [{\citenamefont {Woolley}\ \emph {et~al.}(2008)\citenamefont
  {Woolley}, \citenamefont {Milburn},\ and\ \citenamefont
  {Caves}}]{Woolley_2008}%
  \BibitemOpen
  \bibfield  {author} {\bibinfo {author} {\bibfnamefont {M.~J.}\ \bibnamefont
  {Woolley}}, \bibinfo {author} {\bibfnamefont {G.~J.}\ \bibnamefont
  {Milburn}},\ and\ \bibinfo {author} {\bibfnamefont {C.~M.}\ \bibnamefont
  {Caves}},\ }\bibfield  {title} {\bibinfo {title} {Nonlinear quantum metrology
  using coupled nanomechanical resonators},\ }\href
  {https://doi.org/10.1088/1367-2630/10/12/125018} {\bibfield  {journal}
  {\bibinfo  {journal} {New Journal of Physics}\ }\textbf {\bibinfo {volume}
  {10}},\ \bibinfo {pages} {125018} (\bibinfo {year} {2008})}\BibitemShut
  {NoStop}%
\bibitem [{\citenamefont {Rahman}\ \emph {et~al.}(2024)\citenamefont {Rahman},
  \citenamefont {Kladarić}, \citenamefont {Kern}, \citenamefont {Chu},
  \citenamefont {Filip},\ and\ \citenamefont {Fadel}}]{rahman2024}%
  \BibitemOpen
  \bibfield  {author} {\bibinfo {author} {\bibfnamefont {Q.~R.}\ \bibnamefont
  {Rahman}}, \bibinfo {author} {\bibfnamefont {I.}~\bibnamefont {Kladarić}},
  \bibinfo {author} {\bibfnamefont {M.-E.}\ \bibnamefont {Kern}}, \bibinfo
  {author} {\bibfnamefont {Y.}~\bibnamefont {Chu}}, \bibinfo {author}
  {\bibfnamefont {R.}~\bibnamefont {Filip}},\ and\ \bibinfo {author}
  {\bibfnamefont {M.}~\bibnamefont {Fadel}},\ }\href
  {https://arxiv.org/abs/2412.20971} {\bibinfo {title} {{Genuine Quantum
  non-Gaussianity and metrological sensitivity of Fock states prepared in a
  mechanical resonator}}} (\bibinfo {year} {2024}),\ \Eprint
  {https://arxiv.org/abs/2412.20971} {arXiv:2412.20971 [quant-ph]} \BibitemShut
  {NoStop}%
\bibitem [{\citenamefont {Grishchuk}\ and\ \citenamefont
  {Sazhin}(1983)}]{osti_5123931}%
  \BibitemOpen
  \bibfield  {author} {\bibinfo {author} {\bibfnamefont {L.~P.}\ \bibnamefont
  {Grishchuk}}\ and\ \bibinfo {author} {\bibfnamefont {M.~V.}\ \bibnamefont
  {Sazhin}},\ }\bibfield  {title} {\bibinfo {title} {Squeezed quantum states of
  a harmonic oscillator in the problem of gravitational-wave detection},\
  }\bibfield  {journal} {\bibinfo  {journal} {Sov. Phys. - JETP (Engl.
  Transl.); (United States)}\ }\textbf {\bibinfo {volume} {57:6}},\ \href
  {https://www.osti.gov/biblio/5123931} {} (\bibinfo {year} {1983})\BibitemShut
  {NoStop}%
\bibitem [{\citenamefont {Whittle}\ \emph {et~al.}(2021)\citenamefont {Whittle}
  \emph {et~al.}}]{lll}%
  \BibitemOpen
  \bibfield  {author} {\bibinfo {author} {\bibfnamefont {C.}~\bibnamefont
  {Whittle}} \emph {et~al.},\ }\bibfield  {title} {\bibinfo {title}
  {Approaching the motional ground state of a 10-kg object},\ }\href
  {https://doi.org/10.1126/science.abh2634} {\bibfield  {journal} {\bibinfo
  {journal} {Science}\ }\textbf {\bibinfo {volume} {372}},\ \bibinfo {pages}
  {1333} (\bibinfo {year} {2021})}\BibitemShut {NoStop}%
\bibitem [{\citenamefont {Khalili}\ \emph {et~al.}(2010)\citenamefont
  {Khalili}, \citenamefont {Danilishin}, \citenamefont {Miao}, \citenamefont
  {M\"uller-Ebhardt}, \citenamefont {Yang},\ and\ \citenamefont
  {Chen}}]{PhysRevLett.105.070403}%
  \BibitemOpen
  \bibfield  {author} {\bibinfo {author} {\bibfnamefont {F.}~\bibnamefont
  {Khalili}}, \bibinfo {author} {\bibfnamefont {S.}~\bibnamefont {Danilishin}},
  \bibinfo {author} {\bibfnamefont {H.}~\bibnamefont {Miao}}, \bibinfo {author}
  {\bibfnamefont {H.}~\bibnamefont {M\"uller-Ebhardt}}, \bibinfo {author}
  {\bibfnamefont {H.}~\bibnamefont {Yang}},\ and\ \bibinfo {author}
  {\bibfnamefont {Y.}~\bibnamefont {Chen}},\ }\bibfield  {title} {\bibinfo
  {title} {{Preparing a Mechanical Oscillator in Non-Gaussian Quantum
  States}},\ }\href {https://doi.org/10.1103/PhysRevLett.105.070403} {\bibfield
   {journal} {\bibinfo  {journal} {Phys. Rev. Lett.}\ }\textbf {\bibinfo
  {volume} {105}},\ \bibinfo {pages} {070403} (\bibinfo {year}
  {2010})}\BibitemShut {NoStop}%
\bibitem [{\citenamefont {Mandel}\ and\ \citenamefont
  {Wolf}(1995)}]{Mandel_Wolf_1995}%
  \BibitemOpen
  \bibfield  {author} {\bibinfo {author} {\bibfnamefont {L.}~\bibnamefont
  {Mandel}}\ and\ \bibinfo {author} {\bibfnamefont {E.}~\bibnamefont {Wolf}},\
  }\href@noop {} {\emph {\bibinfo {title} {Optical Coherence and Quantum
  Optics}}}\ (\bibinfo  {publisher} {Cambridge University Press},\ \bibinfo
  {year} {1995})\BibitemShut {NoStop}%
\bibitem [{\citenamefont {Holevo}(1982)}]{holevo}%
  \BibitemOpen
  \bibfield  {author} {\bibinfo {author} {\bibfnamefont {A.~S.}\ \bibnamefont
  {Holevo}},\ }\href {https://doi.org/10.1007/978-88-7642-378-9} {\emph
  {\bibinfo {title} {{Probabilistic and Statistical Aspects of Quantum
  Theory}}}}\ (\bibinfo  {publisher} {North-Holland, Amsterdam},\ \bibinfo
  {year} {1982})\BibitemShut {NoStop}%
\bibitem [{\citenamefont {Chiribella}\ \emph {et~al.}(2006)\citenamefont
  {Chiribella}, \citenamefont {D'Ariano},\ and\ \citenamefont
  {Sacchi}}]{PhysRevA.73.062103}%
  \BibitemOpen
  \bibfield  {author} {\bibinfo {author} {\bibfnamefont {G.}~\bibnamefont
  {Chiribella}}, \bibinfo {author} {\bibfnamefont {G.~M.}\ \bibnamefont
  {D'Ariano}},\ and\ \bibinfo {author} {\bibfnamefont {M.~F.}\ \bibnamefont
  {Sacchi}},\ }\bibfield  {title} {\bibinfo {title} {Optimal estimation of
  squeezing},\ }\href {https://doi.org/10.1103/PhysRevA.73.062103} {\bibfield
  {journal} {\bibinfo  {journal} {Phys. Rev. A}\ }\textbf {\bibinfo {volume}
  {73}},\ \bibinfo {pages} {062103} (\bibinfo {year} {2006})}\BibitemShut
  {NoStop}%
\bibitem [{\citenamefont {Sakurai}(1993)}]{Sakurai1993Modern}%
  \BibitemOpen
  \bibfield  {author} {\bibinfo {author} {\bibfnamefont {J.~J.}\ \bibnamefont
  {Sakurai}},\ }\href {http://www.worldcat.org/isbn/0201539292} {\emph
  {\bibinfo {title} {Modern Quantum Mechanics (Revised Edition)}}}\ (\bibinfo
  {publisher} {Addison Wesley},\ \bibinfo {year} {1993})\BibitemShut {NoStop}%
\bibitem [{\citenamefont {Braunstein}\ and\ \citenamefont
  {Caves}(1994)}]{PhysRevLett.72.3439}%
  \BibitemOpen
  \bibfield  {author} {\bibinfo {author} {\bibfnamefont {S.~L.}\ \bibnamefont
  {Braunstein}}\ and\ \bibinfo {author} {\bibfnamefont {C.~M.}\ \bibnamefont
  {Caves}},\ }\bibfield  {title} {\bibinfo {title} {Statistical distance and
  the geometry of quantum states},\ }\href
  {https://doi.org/10.1103/PhysRevLett.72.3439} {\bibfield  {journal} {\bibinfo
   {journal} {Phys. Rev. Lett.}\ }\textbf {\bibinfo {volume} {72}},\ \bibinfo
  {pages} {3439} (\bibinfo {year} {1994})}\BibitemShut {NoStop}%
\bibitem [{\citenamefont {Faist}\ \emph {et~al.}(2023)\citenamefont {Faist},
  \citenamefont {Woods}, \citenamefont {Albert}, \citenamefont {Renes},
  \citenamefont {Eisert},\ and\ \citenamefont
  {Preskill}}]{PRXQuantum.4.040336}%
  \BibitemOpen
  \bibfield  {author} {\bibinfo {author} {\bibfnamefont {P.}~\bibnamefont
  {Faist}}, \bibinfo {author} {\bibfnamefont {M.~P.}\ \bibnamefont {Woods}},
  \bibinfo {author} {\bibfnamefont {V.~V.}\ \bibnamefont {Albert}}, \bibinfo
  {author} {\bibfnamefont {J.~M.}\ \bibnamefont {Renes}}, \bibinfo {author}
  {\bibfnamefont {J.}~\bibnamefont {Eisert}},\ and\ \bibinfo {author}
  {\bibfnamefont {J.}~\bibnamefont {Preskill}},\ }\bibfield  {title} {\bibinfo
  {title} {Time-energy uncertainty relation for noisy quantum metrology},\
  }\href {https://doi.org/10.1103/PRXQuantum.4.040336} {\bibfield  {journal}
  {\bibinfo  {journal} {PRX Quantum}\ }\textbf {\bibinfo {volume} {4}},\
  \bibinfo {pages} {040336} (\bibinfo {year} {2023})}\BibitemShut {NoStop}%
\bibitem [{\citenamefont {Kholevo}(1974)}]{holmom}%
  \BibitemOpen
  \bibfield  {author} {\bibinfo {author} {\bibfnamefont {A.~S.}\ \bibnamefont
  {Kholevo}},\ }\bibfield  {title} {\bibinfo {title} {{A Generalization of the
  Rao–Cramer Inequality}},\ }\href {https://doi.org/10.1137/1118039}
  {\bibfield  {journal} {\bibinfo  {journal} {Theory of Probability \& Its
  Applications}\ }\textbf {\bibinfo {volume} {18}},\ \bibinfo {pages} {359}
  (\bibinfo {year} {1974})}\BibitemShut {NoStop}%
\bibitem [{\citenamefont {Pezze}\ and\ \citenamefont {Smerzi}()}]{ps}%
  \BibitemOpen
  \bibfield  {author} {\bibinfo {author} {\bibfnamefont {L.}~\bibnamefont
  {Pezze}}\ and\ \bibinfo {author} {\bibfnamefont {A.}~\bibnamefont {Smerzi}},\
  }\bibfield  {title} {\bibinfo {title} {Quantum theory of phase estimation},\
  }in\ \href {https://doi.org/10.3254/978-1-61499-448-0-691} {\emph {\bibinfo
  {booktitle} {Atom Interferometry}}},\ \bibinfo {series} {Proceedings of the
  International School of Physics ``Enrico Ferm''}, Vol.\ \bibinfo {volume}
  {188},\ \bibinfo {editor} {edited by\ \bibinfo {editor} {\bibfnamefont
  {G.~M.}\ \bibnamefont {Tino}}\ and\ \bibinfo {editor} {\bibfnamefont {M.~A.}\
  \bibnamefont {Kasevich}}}\ (\bibinfo  {publisher} {ACM Press},\ \bibinfo
  {address} {New York, NY})\ p.\ \bibinfo {pages} {691}\BibitemShut {NoStop}%
\bibitem [{\citenamefont {Pezz\`e}\ \emph {et~al.}(2018)\citenamefont
  {Pezz\`e}, \citenamefont {Smerzi}, \citenamefont {Oberthaler}, \citenamefont
  {Schmied},\ and\ \citenamefont {Treutlein}}]{RevModPhys.90.035005}%
  \BibitemOpen
  \bibfield  {author} {\bibinfo {author} {\bibfnamefont {L.}~\bibnamefont
  {Pezz\`e}}, \bibinfo {author} {\bibfnamefont {A.}~\bibnamefont {Smerzi}},
  \bibinfo {author} {\bibfnamefont {M.~K.}\ \bibnamefont {Oberthaler}},
  \bibinfo {author} {\bibfnamefont {R.}~\bibnamefont {Schmied}},\ and\ \bibinfo
  {author} {\bibfnamefont {P.}~\bibnamefont {Treutlein}},\ }\bibfield  {title}
  {\bibinfo {title} {Quantum metrology with nonclassical states of atomic
  ensembles},\ }\href {https://doi.org/10.1103/RevModPhys.90.035005} {\bibfield
   {journal} {\bibinfo  {journal} {Rev. Mod. Phys.}\ }\textbf {\bibinfo
  {volume} {90}},\ \bibinfo {pages} {035005} (\bibinfo {year}
  {2018})}\BibitemShut {NoStop}%
\bibitem [{\citenamefont {Yurke}\ \emph {et~al.}(1986)\citenamefont {Yurke},
  \citenamefont {McCall},\ and\ \citenamefont {Klauder}}]{PhysRevA.33.4033}%
  \BibitemOpen
  \bibfield  {author} {\bibinfo {author} {\bibfnamefont {B.}~\bibnamefont
  {Yurke}}, \bibinfo {author} {\bibfnamefont {S.~L.}\ \bibnamefont {McCall}},\
  and\ \bibinfo {author} {\bibfnamefont {J.~R.}\ \bibnamefont {Klauder}},\
  }\bibfield  {title} {\bibinfo {title} {{SU(2) and SU(1,1) interferometers}},\
  }\href {https://doi.org/10.1103/PhysRevA.33.4033} {\bibfield  {journal}
  {\bibinfo  {journal} {Phys. Rev. A}\ }\textbf {\bibinfo {volume} {33}},\
  \bibinfo {pages} {4033} (\bibinfo {year} {1986})}\BibitemShut {NoStop}%
\bibitem [{\citenamefont {Gessner}\ \emph {et~al.}(2019)\citenamefont
  {Gessner}, \citenamefont {Smerzi},\ and\ \citenamefont
  {Pezz\`e}}]{PhysRevLett.122.090503}%
  \BibitemOpen
  \bibfield  {author} {\bibinfo {author} {\bibfnamefont {M.}~\bibnamefont
  {Gessner}}, \bibinfo {author} {\bibfnamefont {A.}~\bibnamefont {Smerzi}},\
  and\ \bibinfo {author} {\bibfnamefont {L.}~\bibnamefont {Pezz\`e}},\
  }\bibfield  {title} {\bibinfo {title} {Metrological nonlinear squeezing
  parameter},\ }\href {https://doi.org/10.1103/PhysRevLett.122.090503}
  {\bibfield  {journal} {\bibinfo  {journal} {Phys. Rev. Lett.}\ }\textbf
  {\bibinfo {volume} {122}},\ \bibinfo {pages} {090503} (\bibinfo {year}
  {2019})}\BibitemShut {NoStop}%
\bibitem [{\citenamefont {Glauber}(1963)}]{PhysRev.130.2529}%
  \BibitemOpen
  \bibfield  {author} {\bibinfo {author} {\bibfnamefont {R.~J.}\ \bibnamefont
  {Glauber}},\ }\bibfield  {title} {\bibinfo {title} {The quantum theory of
  optical coherence},\ }\href {https://doi.org/10.1103/PhysRev.130.2529}
  {\bibfield  {journal} {\bibinfo  {journal} {Phys. Rev.}\ }\textbf {\bibinfo
  {volume} {130}},\ \bibinfo {pages} {2529} (\bibinfo {year}
  {1963})}\BibitemShut {NoStop}%
\bibitem [{\citenamefont {Caves}(1981)}]{PhysRevD.23.1693}%
  \BibitemOpen
  \bibfield  {author} {\bibinfo {author} {\bibfnamefont {C.~M.}\ \bibnamefont
  {Caves}},\ }\bibfield  {title} {\bibinfo {title} {Quantum-mechanical noise in
  an interferometer},\ }\href {https://doi.org/10.1103/PhysRevD.23.1693}
  {\bibfield  {journal} {\bibinfo  {journal} {Phys. Rev. D}\ }\textbf {\bibinfo
  {volume} {23}},\ \bibinfo {pages} {1693} (\bibinfo {year}
  {1981})}\BibitemShut {NoStop}%
\bibitem [{\citenamefont {Ueda}(1989)}]{ueda}%
  \BibitemOpen
  \bibfield  {author} {\bibinfo {author} {\bibfnamefont {M.}~\bibnamefont
  {Ueda}},\ }\bibfield  {title} {\bibinfo {title}
  {{Probability-density-functional description of quantum photodetection
  processes}},\ }\href {https://doi.org/10.1088/0954-8998/1/2/005} {\bibfield
  {journal} {\bibinfo  {journal} {Quantum Optics: Journal of the European
  Optical Society Part B}\ }\textbf {\bibinfo {volume} {1}},\ \bibinfo {pages}
  {131} (\bibinfo {year} {1989})}\BibitemShut {NoStop}%
\bibitem [{\citenamefont {Ferrie}(2014)}]{PhysRevA.90.014101}%
  \BibitemOpen
  \bibfield  {author} {\bibinfo {author} {\bibfnamefont {C.}~\bibnamefont
  {Ferrie}},\ }\bibfield  {title} {\bibinfo {title} {Data-processing
  inequalities for quantum metrology},\ }\href
  {https://doi.org/10.1103/PhysRevA.90.014101} {\bibfield  {journal} {\bibinfo
  {journal} {Phys. Rev. A}\ }\textbf {\bibinfo {volume} {90}},\ \bibinfo
  {pages} {014101} (\bibinfo {year} {2014})}\BibitemShut {NoStop}%
\bibitem [{\citenamefont {De~Palma}\ \emph {et~al.}(2016)\citenamefont
  {De~Palma}, \citenamefont {Trevisan},\ and\ \citenamefont
  {Giovannetti}}]{7442587}%
  \BibitemOpen
  \bibfield  {author} {\bibinfo {author} {\bibfnamefont {G.}~\bibnamefont
  {De~Palma}}, \bibinfo {author} {\bibfnamefont {D.}~\bibnamefont {Trevisan}},\
  and\ \bibinfo {author} {\bibfnamefont {V.}~\bibnamefont {Giovannetti}},\
  }\bibfield  {title} {\bibinfo {title} {Passive states optimize the output of
  {Bosonic} {Gaussian} quantum channels},\ }\href@noop {} {\bibfield  {journal}
  {\bibinfo  {journal} {IEEE Transactions on Information Theory}\ }\textbf
  {\bibinfo {volume} {62}},\ \bibinfo {pages} {2895} (\bibinfo {year}
  {2016})}\BibitemShut {NoStop}%
\bibitem [{\citenamefont {Lu}\ \emph {et~al.}(2010)\citenamefont {Lu},
  \citenamefont {Wang},\ and\ \citenamefont {Sun}}]{PhysRevA.82.042103}%
  \BibitemOpen
  \bibfield  {author} {\bibinfo {author} {\bibfnamefont {X.-M.}\ \bibnamefont
  {Lu}}, \bibinfo {author} {\bibfnamefont {X.}~\bibnamefont {Wang}},\ and\
  \bibinfo {author} {\bibfnamefont {C.~P.}\ \bibnamefont {Sun}},\ }\bibfield
  {title} {\bibinfo {title} {{Quantum Fisher information flow and non-Markovian
  processes of open systems}},\ }\href
  {https://doi.org/10.1103/PhysRevA.82.042103} {\bibfield  {journal} {\bibinfo
  {journal} {Phys. Rev. A}\ }\textbf {\bibinfo {volume} {82}},\ \bibinfo
  {pages} {042103} (\bibinfo {year} {2010})}\BibitemShut {NoStop}%
\bibitem [{\citenamefont {Paris}(2009)}]{paris}%
  \BibitemOpen
  \bibfield  {author} {\bibinfo {author} {\bibfnamefont {M.~G.~A.}\
  \bibnamefont {Paris}},\ }\bibfield  {title} {\bibinfo {title} {{Quantum
  estimation for quantum technology}},\ }\href
  {https://doi.org/10.1142/S0219749909004839} {\bibfield  {journal} {\bibinfo
  {journal} {International Journal of Quantum Information}\ }\textbf {\bibinfo
  {volume} {07}},\ \bibinfo {pages} {125} (\bibinfo {year} {2009})}\BibitemShut
  {NoStop}%
\bibitem [{\citenamefont {\ifmmode~\check{S}\else
  \v{S}\fi{}afr\'anek}(2017)}]{PhysRevA.95.052320}%
  \BibitemOpen
  \bibfield  {author} {\bibinfo {author} {\bibfnamefont {D.}~\bibnamefont
  {\ifmmode~\check{S}\else \v{S}\fi{}afr\'anek}},\ }\bibfield  {title}
  {\bibinfo {title} {{Discontinuities of the quantum Fisher information and the
  Bures metric}},\ }\href {https://doi.org/10.1103/PhysRevA.95.052320}
  {\bibfield  {journal} {\bibinfo  {journal} {Phys. Rev. A}\ }\textbf {\bibinfo
  {volume} {95}},\ \bibinfo {pages} {052320} (\bibinfo {year}
  {2017})}\BibitemShut {NoStop}%
\bibitem [{\citenamefont {Bratteli}\ and\ \citenamefont {Robinson}(1997)}]{br}%
  \BibitemOpen
  \bibfield  {author} {\bibinfo {author} {\bibfnamefont {O.}~\bibnamefont
  {Bratteli}}\ and\ \bibinfo {author} {\bibfnamefont {D.~W.}\ \bibnamefont
  {Robinson}},\ }\href {https://doi.org/10.1007/978-3-662-03444-6} {\emph
  {\bibinfo {title} {{Operator Algebras and Quantum Statiscal Mechanics 2}}}}\
  (\bibinfo  {publisher} {Springer-Verlag Berlin Heidelberg},\ \bibinfo {year}
  {1997})\BibitemShut {NoStop}%
\bibitem [{\citenamefont {Rezakhani}\ \emph {et~al.}(2019)\citenamefont
  {Rezakhani}, \citenamefont {Hassani},\ and\ \citenamefont
  {Alipour}}]{PhysRevA.100.032317}%
  \BibitemOpen
  \bibfield  {author} {\bibinfo {author} {\bibfnamefont {A.~T.}\ \bibnamefont
  {Rezakhani}}, \bibinfo {author} {\bibfnamefont {M.}~\bibnamefont {Hassani}},\
  and\ \bibinfo {author} {\bibfnamefont {S.}~\bibnamefont {Alipour}},\
  }\bibfield  {title} {\bibinfo {title} {{Continuity of the quantum Fisher
  information}},\ }\href {https://doi.org/10.1103/PhysRevA.100.032317}
  {\bibfield  {journal} {\bibinfo  {journal} {Phys. Rev. A}\ }\textbf {\bibinfo
  {volume} {100}},\ \bibinfo {pages} {032317} (\bibinfo {year}
  {2019})}\BibitemShut {NoStop}%
\bibitem [{\citenamefont {Fujiwara}(2001)}]{PhysRevA.63.042304}%
  \BibitemOpen
  \bibfield  {author} {\bibinfo {author} {\bibfnamefont {A.}~\bibnamefont
  {Fujiwara}},\ }\bibfield  {title} {\bibinfo {title} {Quantum channel
  identification problem},\ }\href {https://doi.org/10.1103/PhysRevA.63.042304}
  {\bibfield  {journal} {\bibinfo  {journal} {Phys. Rev. A}\ }\textbf {\bibinfo
  {volume} {63}},\ \bibinfo {pages} {042304} (\bibinfo {year}
  {2001})}\BibitemShut {NoStop}%
\bibitem [{\citenamefont {Rubio}(2022)}]{rubio2022quantum}%
  \BibitemOpen
  \bibfield  {author} {\bibinfo {author} {\bibfnamefont {J.}~\bibnamefont
  {Rubio}},\ }\bibfield  {title} {\bibinfo {title} {Quantum scale estimation},\
  }\href@noop {} {\bibfield  {journal} {\bibinfo  {journal} {Quantum Science
  and Technology}\ }\textbf {\bibinfo {volume} {8}},\ \bibinfo {pages} {015009}
  (\bibinfo {year} {2022})}\BibitemShut {NoStop}%
\bibitem [{\citenamefont {Rubio}\ \emph {et~al.}(2021)\citenamefont {Rubio},
  \citenamefont {Anders},\ and\ \citenamefont
  {Correa}}]{PhysRevLett.127.190402}%
  \BibitemOpen
  \bibfield  {author} {\bibinfo {author} {\bibfnamefont {J.}~\bibnamefont
  {Rubio}}, \bibinfo {author} {\bibfnamefont {J.}~\bibnamefont {Anders}},\ and\
  \bibinfo {author} {\bibfnamefont {L.~A.}\ \bibnamefont {Correa}},\ }\bibfield
   {title} {\bibinfo {title} {Global quantum thermometry},\ }\href
  {https://doi.org/10.1103/PhysRevLett.127.190402} {\bibfield  {journal}
  {\bibinfo  {journal} {Phys. Rev. Lett.}\ }\textbf {\bibinfo {volume} {127}},\
  \bibinfo {pages} {190402} (\bibinfo {year} {2021})}\BibitemShut {NoStop}%
\bibitem [{\citenamefont {Barndorff-Nielsen}\ \emph {et~al.}(2003)\citenamefont
  {Barndorff-Nielsen}, \citenamefont {Gill},\ and\ \citenamefont {Jupp}}]{bng}%
  \BibitemOpen
  \bibfield  {author} {\bibinfo {author} {\bibfnamefont {O.~E.}\ \bibnamefont
  {Barndorff-Nielsen}}, \bibinfo {author} {\bibfnamefont {R.~D.}\ \bibnamefont
  {Gill}},\ and\ \bibinfo {author} {\bibfnamefont {P.~E.}\ \bibnamefont
  {Jupp}},\ }\bibfield  {title} {\bibinfo {title} {On quantum statistical
  inference},\ }\href
  {https://rss.onlinelibrary.wiley.com/doi/abs/10.1111/1467-9868.00415}
  {\bibfield  {journal} {\bibinfo  {journal} {Journal of the Royal Statistical
  Society: Series B (Statistical Methodology)}\ }\textbf {\bibinfo {volume}
  {65}},\ \bibinfo {pages} {775} (\bibinfo {year} {2003})}\BibitemShut
  {NoStop}%
\bibitem [{\citenamefont {Lehmann}\ and\ \citenamefont
  {Casella}(1998)}]{lehmanntheory}%
  \BibitemOpen
  \bibfield  {author} {\bibinfo {author} {\bibfnamefont {E.~L.}\ \bibnamefont
  {Lehmann}}\ and\ \bibinfo {author} {\bibfnamefont {G.}~\bibnamefont
  {Casella}},\ }\href@noop {} {\emph {\bibinfo {title} {Theory of Point
  Estimation}}}\ (\bibinfo  {publisher} {Springer},\ \bibinfo {year}
  {1998})\BibitemShut {NoStop}%
\bibitem [{\citenamefont {Van~der Vaart}(2000)}]{van2000asymptotic}%
  \BibitemOpen
  \bibfield  {author} {\bibinfo {author} {\bibfnamefont {A.~W.}\ \bibnamefont
  {Van~der Vaart}},\ }\href@noop {} {\emph {\bibinfo {title} {Asymptotic
  statistics}}},\ Vol.~\bibinfo {volume} {3}\ (\bibinfo  {publisher} {Cambridge
  University Press},\ \bibinfo {year} {2000})\BibitemShut {NoStop}%
\bibitem [{\citenamefont {Shao}(2008)}]{shao2008mathematical}%
  \BibitemOpen
  \bibfield  {author} {\bibinfo {author} {\bibfnamefont {J.}~\bibnamefont
  {Shao}},\ }\href@noop {} {\emph {\bibinfo {title} {Mathematical
  statistics}}}\ (\bibinfo  {publisher} {Springer Science \& Business Media},\
  \bibinfo {year} {2008})\BibitemShut {NoStop}%
\bibitem [{\citenamefont {Humphreys}\ \emph {et~al.}(2013)\citenamefont
  {Humphreys}, \citenamefont {Barbieri}, \citenamefont {Datta},\ and\
  \citenamefont {Walmsley}}]{PhysRevLett.111.070403}%
  \BibitemOpen
  \bibfield  {author} {\bibinfo {author} {\bibfnamefont {P.~C.}\ \bibnamefont
  {Humphreys}}, \bibinfo {author} {\bibfnamefont {M.}~\bibnamefont {Barbieri}},
  \bibinfo {author} {\bibfnamefont {A.}~\bibnamefont {Datta}},\ and\ \bibinfo
  {author} {\bibfnamefont {I.~A.}\ \bibnamefont {Walmsley}},\ }\bibfield
  {title} {\bibinfo {title} {Quantum enhanced multiple phase estimation},\
  }\href {https://doi.org/10.1103/PhysRevLett.111.070403} {\bibfield  {journal}
  {\bibinfo  {journal} {Phys. Rev. Lett.}\ }\textbf {\bibinfo {volume} {111}},\
  \bibinfo {pages} {070403} (\bibinfo {year} {2013})}\BibitemShut {NoStop}%
\bibitem [{\citenamefont {Volkoff}\ and\ \citenamefont
  {Whaley}(2014)}]{PhysRevA.89.012122}%
  \BibitemOpen
  \bibfield  {author} {\bibinfo {author} {\bibfnamefont {T.~J.}\ \bibnamefont
  {Volkoff}}\ and\ \bibinfo {author} {\bibfnamefont {K.~B.}\ \bibnamefont
  {Whaley}},\ }\bibfield  {title} {\bibinfo {title} {{Measurement- and
  comparison-based sizes of Schr\"odinger cat states of light}},\ }\href
  {https://doi.org/10.1103/PhysRevA.89.012122} {\bibfield  {journal} {\bibinfo
  {journal} {Phys. Rev. A}\ }\textbf {\bibinfo {volume} {89}},\ \bibinfo
  {pages} {012122} (\bibinfo {year} {2014})}\BibitemShut {NoStop}%
\bibitem [{\citenamefont {Ge}\ \emph {et~al.}(2019)\citenamefont {Ge},
  \citenamefont {Sawyer}, \citenamefont {Britton}, \citenamefont {Jacobs},
  \citenamefont {Bollinger},\ and\ \citenamefont
  {Foss-Feig}}]{PhysRevLett.122.030501}%
  \BibitemOpen
  \bibfield  {author} {\bibinfo {author} {\bibfnamefont {W.}~\bibnamefont
  {Ge}}, \bibinfo {author} {\bibfnamefont {B.~C.}\ \bibnamefont {Sawyer}},
  \bibinfo {author} {\bibfnamefont {J.~W.}\ \bibnamefont {Britton}}, \bibinfo
  {author} {\bibfnamefont {K.}~\bibnamefont {Jacobs}}, \bibinfo {author}
  {\bibfnamefont {J.~J.}\ \bibnamefont {Bollinger}},\ and\ \bibinfo {author}
  {\bibfnamefont {M.}~\bibnamefont {Foss-Feig}},\ }\bibfield  {title} {\bibinfo
  {title} {Trapped ion quantum information processing with squeezed phonons},\
  }\href {https://doi.org/10.1103/PhysRevLett.122.030501} {\bibfield  {journal}
  {\bibinfo  {journal} {Phys. Rev. Lett.}\ }\textbf {\bibinfo {volume} {122}},\
  \bibinfo {pages} {030501} (\bibinfo {year} {2019})}\BibitemShut {NoStop}%
\bibitem [{\citenamefont {M\o{}lmer}\ and\ \citenamefont
  {S\o{}rensen}(1999)}]{PhysRevLett.82.1835}%
  \BibitemOpen
  \bibfield  {author} {\bibinfo {author} {\bibfnamefont {K.}~\bibnamefont
  {M\o{}lmer}}\ and\ \bibinfo {author} {\bibfnamefont {A.}~\bibnamefont
  {S\o{}rensen}},\ }\bibfield  {title} {\bibinfo {title} {Multiparticle
  entanglement of hot trapped ions},\ }\href
  {https://doi.org/10.1103/PhysRevLett.82.1835} {\bibfield  {journal} {\bibinfo
   {journal} {Phys. Rev. Lett.}\ }\textbf {\bibinfo {volume} {82}},\ \bibinfo
  {pages} {1835} (\bibinfo {year} {1999})}\BibitemShut {NoStop}%
\bibitem [{\citenamefont {Cirac}\ and\ \citenamefont
  {Zoller}(1995)}]{PhysRevLett.74.4091}%
  \BibitemOpen
  \bibfield  {author} {\bibinfo {author} {\bibfnamefont {J.~I.}\ \bibnamefont
  {Cirac}}\ and\ \bibinfo {author} {\bibfnamefont {P.}~\bibnamefont {Zoller}},\
  }\bibfield  {title} {\bibinfo {title} {Quantum computations with cold trapped
  ions},\ }\href {https://doi.org/10.1103/PhysRevLett.74.4091} {\bibfield
  {journal} {\bibinfo  {journal} {Phys. Rev. Lett.}\ }\textbf {\bibinfo
  {volume} {74}},\ \bibinfo {pages} {4091} (\bibinfo {year}
  {1995})}\BibitemShut {NoStop}%
\bibitem [{\citenamefont {Um}\ \emph {et~al.}(2016)\citenamefont {Um},
  \citenamefont {Zhang}, \citenamefont {Lv}, \citenamefont {Lu}, \citenamefont
  {An}, \citenamefont {Zhang}, \citenamefont {Nha}, \citenamefont {Kim},\ and\
  \citenamefont {Kim}}]{Um2016}%
  \BibitemOpen
  \bibfield  {author} {\bibinfo {author} {\bibfnamefont {M.}~\bibnamefont
  {Um}}, \bibinfo {author} {\bibfnamefont {J.}~\bibnamefont {Zhang}}, \bibinfo
  {author} {\bibfnamefont {D.}~\bibnamefont {Lv}}, \bibinfo {author}
  {\bibfnamefont {Y.}~\bibnamefont {Lu}}, \bibinfo {author} {\bibfnamefont
  {S.}~\bibnamefont {An}}, \bibinfo {author} {\bibfnamefont {J.-N.}\
  \bibnamefont {Zhang}}, \bibinfo {author} {\bibfnamefont {H.}~\bibnamefont
  {Nha}}, \bibinfo {author} {\bibfnamefont {M.~S.}\ \bibnamefont {Kim}},\ and\
  \bibinfo {author} {\bibfnamefont {K.}~\bibnamefont {Kim}},\ }\bibfield
  {title} {\bibinfo {title} {Phonon arithmetic in a trapped ion system},\
  }\href {https://doi.org/10.1038/ncomms11410} {\bibfield  {journal} {\bibinfo
  {journal} {Nature Communications}\ }\textbf {\bibinfo {volume} {7}},\
  \bibinfo {pages} {11410} (\bibinfo {year} {2016})}\BibitemShut {NoStop}%
\bibitem [{\citenamefont {Wolf}\ \emph {et~al.}(2019)\citenamefont {Wolf},
  \citenamefont {Shi}, \citenamefont {Heip}, \citenamefont {Gessner},
  \citenamefont {Pezz{\`e}}, \citenamefont {Smerzi}, \citenamefont {Schulte},
  \citenamefont {Hammerer},\ and\ \citenamefont {Schmidt}}]{Wolf2019}%
  \BibitemOpen
  \bibfield  {author} {\bibinfo {author} {\bibfnamefont {F.}~\bibnamefont
  {Wolf}}, \bibinfo {author} {\bibfnamefont {C.}~\bibnamefont {Shi}}, \bibinfo
  {author} {\bibfnamefont {J.~C.}\ \bibnamefont {Heip}}, \bibinfo {author}
  {\bibfnamefont {M.}~\bibnamefont {Gessner}}, \bibinfo {author} {\bibfnamefont
  {L.}~\bibnamefont {Pezz{\`e}}}, \bibinfo {author} {\bibfnamefont
  {A.}~\bibnamefont {Smerzi}}, \bibinfo {author} {\bibfnamefont
  {M.}~\bibnamefont {Schulte}}, \bibinfo {author} {\bibfnamefont
  {K.}~\bibnamefont {Hammerer}},\ and\ \bibinfo {author} {\bibfnamefont
  {P.~O.}\ \bibnamefont {Schmidt}},\ }\bibfield  {title} {\bibinfo {title}
  {{Motional Fock states for quantum-enhanced amplitude and phase measurements
  with trapped ions}},\ }\href {https://doi.org/10.1038/s41467-019-10576-4}
  {\bibfield  {journal} {\bibinfo  {journal} {Nature Communications}\ }\textbf
  {\bibinfo {volume} {10}},\ \bibinfo {pages} {2929} (\bibinfo {year}
  {2019})}\BibitemShut {NoStop}%
\bibitem [{\citenamefont {Fl\"uhmann}\ and\ \citenamefont
  {Home}(2020)}]{PhysRevLett.125.043602}%
  \BibitemOpen
  \bibfield  {author} {\bibinfo {author} {\bibfnamefont {C.}~\bibnamefont
  {Fl\"uhmann}}\ and\ \bibinfo {author} {\bibfnamefont {J.~P.}\ \bibnamefont
  {Home}},\ }\bibfield  {title} {\bibinfo {title} {{Direct
  Characteristic-Function Tomography of Quantum States of the Trapped-Ion
  Motional Oscillator}},\ }\href
  {https://doi.org/10.1103/PhysRevLett.125.043602} {\bibfield  {journal}
  {\bibinfo  {journal} {Phys. Rev. Lett.}\ }\textbf {\bibinfo {volume} {125}},\
  \bibinfo {pages} {043602} (\bibinfo {year} {2020})}\BibitemShut {NoStop}%
\bibitem [{\citenamefont {Mallweger}\ \emph {et~al.}(2023)\citenamefont
  {Mallweger}, \citenamefont {de~Oliveira}, \citenamefont {Thomm},
  \citenamefont {Parke}, \citenamefont {Kuk}, \citenamefont {Higgins},
  \citenamefont {Bachelard}, \citenamefont {Villas-Boas},\ and\ \citenamefont
  {Hennrich}}]{PhysRevLett.131.223603}%
  \BibitemOpen
  \bibfield  {author} {\bibinfo {author} {\bibfnamefont {M.}~\bibnamefont
  {Mallweger}}, \bibinfo {author} {\bibfnamefont {M.~H.}\ \bibnamefont
  {de~Oliveira}}, \bibinfo {author} {\bibfnamefont {R.}~\bibnamefont {Thomm}},
  \bibinfo {author} {\bibfnamefont {H.}~\bibnamefont {Parke}}, \bibinfo
  {author} {\bibfnamefont {N.}~\bibnamefont {Kuk}}, \bibinfo {author}
  {\bibfnamefont {G.}~\bibnamefont {Higgins}}, \bibinfo {author} {\bibfnamefont
  {R.}~\bibnamefont {Bachelard}}, \bibinfo {author} {\bibfnamefont {C.~J.}\
  \bibnamefont {Villas-Boas}},\ and\ \bibinfo {author} {\bibfnamefont
  {M.}~\bibnamefont {Hennrich}},\ }\bibfield  {title} {\bibinfo {title}
  {{Single-Shot Measurements of Phonon Number States Using the Autler-Townes
  Effect}},\ }\href {https://doi.org/10.1103/PhysRevLett.131.223603} {\bibfield
   {journal} {\bibinfo  {journal} {Phys. Rev. Lett.}\ }\textbf {\bibinfo
  {volume} {131}},\ \bibinfo {pages} {223603} (\bibinfo {year}
  {2023})}\BibitemShut {NoStop}%
\bibitem [{\citenamefont {Katz}\ \emph {et~al.}(2023)\citenamefont {Katz},
  \citenamefont {Cetina},\ and\ \citenamefont {Monroe}}]{PRXQuantum.4.030311}%
  \BibitemOpen
  \bibfield  {author} {\bibinfo {author} {\bibfnamefont {O.}~\bibnamefont
  {Katz}}, \bibinfo {author} {\bibfnamefont {M.}~\bibnamefont {Cetina}},\ and\
  \bibinfo {author} {\bibfnamefont {C.}~\bibnamefont {Monroe}},\ }\bibfield
  {title} {\bibinfo {title} {{Programmable $N$-Body Interactions with Trapped
  Ions}},\ }\href {https://doi.org/10.1103/PRXQuantum.4.030311} {\bibfield
  {journal} {\bibinfo  {journal} {PRX Quantum}\ }\textbf {\bibinfo {volume}
  {4}},\ \bibinfo {pages} {030311} (\bibinfo {year} {2023})}\BibitemShut
  {NoStop}%
\bibitem [{\citenamefont {Lewis-Swan}\ \emph {et~al.}(2024)\citenamefont
  {Lewis-Swan}, \citenamefont {Castro}, \citenamefont {Barberena},\ and\
  \citenamefont {Rey}}]{PhysRevLett.132.163601}%
  \BibitemOpen
  \bibfield  {author} {\bibinfo {author} {\bibfnamefont {R.~J.}\ \bibnamefont
  {Lewis-Swan}}, \bibinfo {author} {\bibfnamefont {J.~C.~Z.}\ \bibnamefont
  {Castro}}, \bibinfo {author} {\bibfnamefont {D.}~\bibnamefont {Barberena}},\
  and\ \bibinfo {author} {\bibfnamefont {A.~M.}\ \bibnamefont {Rey}},\
  }\bibfield  {title} {\bibinfo {title} {Exploiting nonclassical motion of a
  trapped ion crystal for quantum-enhanced metrology of global and differential
  spin rotations},\ }\href {https://doi.org/10.1103/PhysRevLett.132.163601}
  {\bibfield  {journal} {\bibinfo  {journal} {Phys. Rev. Lett.}\ }\textbf
  {\bibinfo {volume} {132}},\ \bibinfo {pages} {163601} (\bibinfo {year}
  {2024})}\BibitemShut {NoStop}%
\bibitem [{\citenamefont {Katz}\ \emph {et~al.}(2022)\citenamefont {Katz},
  \citenamefont {Cetina},\ and\ \citenamefont
  {Monroe}}]{PhysRevLett.129.063603}%
  \BibitemOpen
  \bibfield  {author} {\bibinfo {author} {\bibfnamefont {O.}~\bibnamefont
  {Katz}}, \bibinfo {author} {\bibfnamefont {M.}~\bibnamefont {Cetina}},\ and\
  \bibinfo {author} {\bibfnamefont {C.}~\bibnamefont {Monroe}},\ }\bibfield
  {title} {\bibinfo {title} {{$N$-Body Interactions between Trapped Ion Qubits
  via Spin-Dependent Squeezing}},\ }\href
  {https://doi.org/10.1103/PhysRevLett.129.063603} {\bibfield  {journal}
  {\bibinfo  {journal} {Phys. Rev. Lett.}\ }\textbf {\bibinfo {volume} {129}},\
  \bibinfo {pages} {063603} (\bibinfo {year} {2022})}\BibitemShut {NoStop}%
\bibitem [{\citenamefont {Dylewsky}\ \emph {et~al.}(2016)\citenamefont
  {Dylewsky}, \citenamefont {Freericks}, \citenamefont {Wall}, \citenamefont
  {Rey},\ and\ \citenamefont {Foss-Feig}}]{PhysRevA.93.013415}%
  \BibitemOpen
  \bibfield  {author} {\bibinfo {author} {\bibfnamefont {D.}~\bibnamefont
  {Dylewsky}}, \bibinfo {author} {\bibfnamefont {J.~K.}\ \bibnamefont
  {Freericks}}, \bibinfo {author} {\bibfnamefont {M.~L.}\ \bibnamefont {Wall}},
  \bibinfo {author} {\bibfnamefont {A.~M.}\ \bibnamefont {Rey}},\ and\ \bibinfo
  {author} {\bibfnamefont {M.}~\bibnamefont {Foss-Feig}},\ }\bibfield  {title}
  {\bibinfo {title} {Nonperturbative calculation of phonon effects on spin
  squeezing},\ }\href {https://doi.org/10.1103/PhysRevA.93.013415} {\bibfield
  {journal} {\bibinfo  {journal} {Phys. Rev. A}\ }\textbf {\bibinfo {volume}
  {93}},\ \bibinfo {pages} {013415} (\bibinfo {year} {2016})}\BibitemShut
  {NoStop}%
\bibitem [{\citenamefont {Burd}\ \emph {et~al.}(2019)\citenamefont {Burd},
  \citenamefont {Srinivas}, \citenamefont {Bollinger}, \citenamefont {Wilson},
  \citenamefont {Wineland}, \citenamefont {Leibfried}, \citenamefont
  {Slichter},\ and\ \citenamefont {Allcock}}]{wineland}%
  \BibitemOpen
  \bibfield  {author} {\bibinfo {author} {\bibfnamefont {S.~C.}\ \bibnamefont
  {Burd}}, \bibinfo {author} {\bibfnamefont {R.}~\bibnamefont {Srinivas}},
  \bibinfo {author} {\bibfnamefont {J.~J.}\ \bibnamefont {Bollinger}}, \bibinfo
  {author} {\bibfnamefont {A.~C.}\ \bibnamefont {Wilson}}, \bibinfo {author}
  {\bibfnamefont {D.~J.}\ \bibnamefont {Wineland}}, \bibinfo {author}
  {\bibfnamefont {D.}~\bibnamefont {Leibfried}}, \bibinfo {author}
  {\bibfnamefont {D.~H.}\ \bibnamefont {Slichter}},\ and\ \bibinfo {author}
  {\bibfnamefont {D.~T.~C.}\ \bibnamefont {Allcock}},\ }\bibfield  {title}
  {\bibinfo {title} {Quantum amplification of mechanical oscillator motion},\
  }\href {https://doi.org/10.1126/science.aaw2884} {\bibfield  {journal}
  {\bibinfo  {journal} {Science}\ }\textbf {\bibinfo {volume} {364}},\ \bibinfo
  {pages} {1163} (\bibinfo {year} {2019})}\BibitemShut {NoStop}%
\bibitem [{\citenamefont {Affolter}\ \emph {et~al.}(2023)\citenamefont
  {Affolter}, \citenamefont {Ge}, \citenamefont {Bullock}, \citenamefont
  {Burd}, \citenamefont {Gilmore}, \citenamefont {Lilieholm}, \citenamefont
  {Carter},\ and\ \citenamefont {Bollinger}}]{PhysRevA.107.032425}%
  \BibitemOpen
  \bibfield  {author} {\bibinfo {author} {\bibfnamefont {M.}~\bibnamefont
  {Affolter}}, \bibinfo {author} {\bibfnamefont {W.}~\bibnamefont {Ge}},
  \bibinfo {author} {\bibfnamefont {B.}~\bibnamefont {Bullock}}, \bibinfo
  {author} {\bibfnamefont {S.~C.}\ \bibnamefont {Burd}}, \bibinfo {author}
  {\bibfnamefont {K.~A.}\ \bibnamefont {Gilmore}}, \bibinfo {author}
  {\bibfnamefont {J.~F.}\ \bibnamefont {Lilieholm}}, \bibinfo {author}
  {\bibfnamefont {A.~L.}\ \bibnamefont {Carter}},\ and\ \bibinfo {author}
  {\bibfnamefont {J.~J.}\ \bibnamefont {Bollinger}},\ }\bibfield  {title}
  {\bibinfo {title} {Toward improved quantum simulations and sensing with
  trapped two-dimensional ion crystals via parametric amplification},\ }\href
  {https://doi.org/10.1103/PhysRevA.107.032425} {\bibfield  {journal} {\bibinfo
   {journal} {Phys. Rev. A}\ }\textbf {\bibinfo {volume} {107}},\ \bibinfo
  {pages} {032425} (\bibinfo {year} {2023})}\BibitemShut {NoStop}%
\bibitem [{\citenamefont {Hucul}\ \emph {et~al.}(2015)\citenamefont {Hucul},
  \citenamefont {Inlek}, \citenamefont {Vittorini}, \citenamefont {Crocker},
  \citenamefont {Debnath}, \citenamefont {Clark},\ and\ \citenamefont
  {Monroe}}]{Hucul2015}%
  \BibitemOpen
  \bibfield  {author} {\bibinfo {author} {\bibfnamefont {D.}~\bibnamefont
  {Hucul}}, \bibinfo {author} {\bibfnamefont {I.~V.}\ \bibnamefont {Inlek}},
  \bibinfo {author} {\bibfnamefont {G.}~\bibnamefont {Vittorini}}, \bibinfo
  {author} {\bibfnamefont {C.}~\bibnamefont {Crocker}}, \bibinfo {author}
  {\bibfnamefont {S.}~\bibnamefont {Debnath}}, \bibinfo {author} {\bibfnamefont
  {S.~M.}\ \bibnamefont {Clark}},\ and\ \bibinfo {author} {\bibfnamefont
  {C.}~\bibnamefont {Monroe}},\ }\bibfield  {title} {\bibinfo {title} {Modular
  entanglement of atomic qubits using photons and phonons},\ }\href
  {https://doi.org/10.1038/nphys3150} {\bibfield  {journal} {\bibinfo
  {journal} {Nature Physics}\ }\textbf {\bibinfo {volume} {11}},\ \bibinfo
  {pages} {37} (\bibinfo {year} {2015})}\BibitemShut {NoStop}%
\bibitem [{\citenamefont {Hughes}\ \emph {et~al.}(2020)\citenamefont {Hughes},
  \citenamefont {Sch\"afer}, \citenamefont {Thirumalai}, \citenamefont
  {Nadlinger}, \citenamefont {Woodrow}, \citenamefont {Lucas},\ and\
  \citenamefont {Ballance}}]{PhysRevLett.125.080504}%
  \BibitemOpen
  \bibfield  {author} {\bibinfo {author} {\bibfnamefont {A.~C.}\ \bibnamefont
  {Hughes}}, \bibinfo {author} {\bibfnamefont {V.~M.}\ \bibnamefont
  {Sch\"afer}}, \bibinfo {author} {\bibfnamefont {K.}~\bibnamefont
  {Thirumalai}}, \bibinfo {author} {\bibfnamefont {D.~P.}\ \bibnamefont
  {Nadlinger}}, \bibinfo {author} {\bibfnamefont {S.~R.}\ \bibnamefont
  {Woodrow}}, \bibinfo {author} {\bibfnamefont {D.~M.}\ \bibnamefont {Lucas}},\
  and\ \bibinfo {author} {\bibfnamefont {C.~J.}\ \bibnamefont {Ballance}},\
  }\bibfield  {title} {\bibinfo {title} {Benchmarking a high-fidelity
  mixed-species entangling gate},\ }\href
  {https://doi.org/10.1103/PhysRevLett.125.080504} {\bibfield  {journal}
  {\bibinfo  {journal} {Phys. Rev. Lett.}\ }\textbf {\bibinfo {volume} {125}},\
  \bibinfo {pages} {080504} (\bibinfo {year} {2020})}\BibitemShut {NoStop}%
\bibitem [{\citenamefont {Clarke}\ \emph {et~al.}(2020)\citenamefont {Clarke},
  \citenamefont {Sahium}, \citenamefont {Khosla}, \citenamefont {Pikovski},
  \citenamefont {Kim},\ and\ \citenamefont {Vanner}}]{Clarke_2020}%
  \BibitemOpen
  \bibfield  {author} {\bibinfo {author} {\bibfnamefont {J.}~\bibnamefont
  {Clarke}}, \bibinfo {author} {\bibfnamefont {P.}~\bibnamefont {Sahium}},
  \bibinfo {author} {\bibfnamefont {K.~E.}\ \bibnamefont {Khosla}}, \bibinfo
  {author} {\bibfnamefont {I.}~\bibnamefont {Pikovski}}, \bibinfo {author}
  {\bibfnamefont {M.~S.}\ \bibnamefont {Kim}},\ and\ \bibinfo {author}
  {\bibfnamefont {M.~R.}\ \bibnamefont {Vanner}},\ }\bibfield  {title}
  {\bibinfo {title} {Generating mechanical and optomechanical entanglement via
  pulsed interaction and measurement},\ }\href
  {https://doi.org/10.1088/1367-2630/ab7ddd} {\bibfield  {journal} {\bibinfo
  {journal} {New Journal of Physics}\ }\textbf {\bibinfo {volume} {22}},\
  \bibinfo {pages} {063001} (\bibinfo {year} {2020})}\BibitemShut {NoStop}%
\bibitem [{\citenamefont {de~Moraes~Neto}\ \emph {et~al.}(2019)\citenamefont
  {de~Moraes~Neto}, \citenamefont {Montenegro}, \citenamefont {Teizen},\ and\
  \citenamefont {Vernek}}]{PhysRevA.99.043836}%
  \BibitemOpen
  \bibfield  {author} {\bibinfo {author} {\bibfnamefont {G.~D.}\ \bibnamefont
  {de~Moraes~Neto}}, \bibinfo {author} {\bibfnamefont {V.}~\bibnamefont
  {Montenegro}}, \bibinfo {author} {\bibfnamefont {V.~F.}\ \bibnamefont
  {Teizen}},\ and\ \bibinfo {author} {\bibfnamefont {E.}~\bibnamefont
  {Vernek}},\ }\bibfield  {title} {\bibinfo {title} {Dissipative
  {phonon-Fock-state} production in strong nonlinear optomechanics},\ }\href
  {https://doi.org/10.1103/PhysRevA.99.043836} {\bibfield  {journal} {\bibinfo
  {journal} {Phys. Rev. A}\ }\textbf {\bibinfo {volume} {99}},\ \bibinfo
  {pages} {043836} (\bibinfo {year} {2019})}\BibitemShut {NoStop}%
\bibitem [{\citenamefont {Gottesman}\ \emph {et~al.}(2001)\citenamefont
  {Gottesman}, \citenamefont {Kitaev},\ and\ \citenamefont
  {Preskill}}]{PhysRevA.64.012310}%
  \BibitemOpen
  \bibfield  {author} {\bibinfo {author} {\bibfnamefont {D.}~\bibnamefont
  {Gottesman}}, \bibinfo {author} {\bibfnamefont {A.}~\bibnamefont {Kitaev}},\
  and\ \bibinfo {author} {\bibfnamefont {J.}~\bibnamefont {Preskill}},\
  }\bibfield  {title} {\bibinfo {title} {Encoding a qubit in an oscillator},\
  }\href {https://doi.org/10.1103/PhysRevA.64.012310} {\bibfield  {journal}
  {\bibinfo  {journal} {Phys. Rev. A}\ }\textbf {\bibinfo {volume} {64}},\
  \bibinfo {pages} {012310} (\bibinfo {year} {2001})}\BibitemShut {NoStop}%
\bibitem [{\citenamefont {Matsos}\ \emph {et~al.}(2024)\citenamefont {Matsos},
  \citenamefont {Valahu}, \citenamefont {Navickas}, \citenamefont {Rao},
  \citenamefont {Millican}, \citenamefont {Kolesnikow}, \citenamefont
  {Biercuk},\ and\ \citenamefont {Tan}}]{PhysRevLett.133.050602}%
  \BibitemOpen
  \bibfield  {author} {\bibinfo {author} {\bibfnamefont {V.~G.}\ \bibnamefont
  {Matsos}}, \bibinfo {author} {\bibfnamefont {C.~H.}\ \bibnamefont {Valahu}},
  \bibinfo {author} {\bibfnamefont {T.}~\bibnamefont {Navickas}}, \bibinfo
  {author} {\bibfnamefont {A.~D.}\ \bibnamefont {Rao}}, \bibinfo {author}
  {\bibfnamefont {M.~J.}\ \bibnamefont {Millican}}, \bibinfo {author}
  {\bibfnamefont {X.~C.}\ \bibnamefont {Kolesnikow}}, \bibinfo {author}
  {\bibfnamefont {M.~J.}\ \bibnamefont {Biercuk}},\ and\ \bibinfo {author}
  {\bibfnamefont {T.~R.}\ \bibnamefont {Tan}},\ }\bibfield  {title} {\bibinfo
  {title} {{Robust and Deterministic Preparation of Bosonic Logical States in a
  Trapped Ion}},\ }\href {https://doi.org/10.1103/PhysRevLett.133.050602}
  {\bibfield  {journal} {\bibinfo  {journal} {Phys. Rev. Lett.}\ }\textbf
  {\bibinfo {volume} {133}},\ \bibinfo {pages} {050602} (\bibinfo {year}
  {2024})}\BibitemShut {NoStop}%
\bibitem [{\citenamefont {Burd}\ \emph {et~al.}(2024)\citenamefont {Burd},
  \citenamefont {Knaack}, \citenamefont {Srinivas}, \citenamefont {Arenz},
  \citenamefont {Collopy}, \citenamefont {Stephenson}, \citenamefont {Wilson},
  \citenamefont {Wineland}, \citenamefont {Leibfried}, \citenamefont
  {Bollinger}, \citenamefont {Allcock},\ and\ \citenamefont
  {Slichter}}]{PRXQuantum.5.020314}%
  \BibitemOpen
  \bibfield  {author} {\bibinfo {author} {\bibfnamefont {S.~C.}\ \bibnamefont
  {Burd}}, \bibinfo {author} {\bibfnamefont {H.~M.}\ \bibnamefont {Knaack}},
  \bibinfo {author} {\bibfnamefont {R.}~\bibnamefont {Srinivas}}, \bibinfo
  {author} {\bibfnamefont {C.}~\bibnamefont {Arenz}}, \bibinfo {author}
  {\bibfnamefont {A.~L.}\ \bibnamefont {Collopy}}, \bibinfo {author}
  {\bibfnamefont {L.~J.}\ \bibnamefont {Stephenson}}, \bibinfo {author}
  {\bibfnamefont {A.~C.}\ \bibnamefont {Wilson}}, \bibinfo {author}
  {\bibfnamefont {D.~J.}\ \bibnamefont {Wineland}}, \bibinfo {author}
  {\bibfnamefont {D.}~\bibnamefont {Leibfried}}, \bibinfo {author}
  {\bibfnamefont {J.~J.}\ \bibnamefont {Bollinger}}, \bibinfo {author}
  {\bibfnamefont {D.~T.~C.}\ \bibnamefont {Allcock}},\ and\ \bibinfo {author}
  {\bibfnamefont {D.~H.}\ \bibnamefont {Slichter}},\ }\bibfield  {title}
  {\bibinfo {title} {Experimental speedup of quantum dynamics through
  squeezing},\ }\href {https://doi.org/10.1103/PRXQuantum.5.020314} {\bibfield
  {journal} {\bibinfo  {journal} {PRX Quantum}\ }\textbf {\bibinfo {volume}
  {5}},\ \bibinfo {pages} {020314} (\bibinfo {year} {2024})}\BibitemShut
  {NoStop}%
\bibitem [{\citenamefont {Proctor}\ \emph {et~al.}(2018)\citenamefont
  {Proctor}, \citenamefont {Knott},\ and\ \citenamefont
  {Dunningham}}]{PhysRevLett.120.080501}%
  \BibitemOpen
  \bibfield  {author} {\bibinfo {author} {\bibfnamefont {T.~J.}\ \bibnamefont
  {Proctor}}, \bibinfo {author} {\bibfnamefont {P.~A.}\ \bibnamefont {Knott}},\
  and\ \bibinfo {author} {\bibfnamefont {J.~A.}\ \bibnamefont {Dunningham}},\
  }\bibfield  {title} {\bibinfo {title} {Multiparameter estimation in networked
  quantum sensors},\ }\href {https://doi.org/10.1103/PhysRevLett.120.080501}
  {\bibfield  {journal} {\bibinfo  {journal} {Phys. Rev. Lett.}\ }\textbf
  {\bibinfo {volume} {120}},\ \bibinfo {pages} {080501} (\bibinfo {year}
  {2018})}\BibitemShut {NoStop}%
\bibitem [{\citenamefont {Fujiwara}(2006)}]{fujiwara}%
  \BibitemOpen
  \bibfield  {author} {\bibinfo {author} {\bibfnamefont {A.}~\bibnamefont
  {Fujiwara}},\ }\bibfield  {title} {\bibinfo {title} {{Strong consistency and
  asymptotic efficiency for adaptive quantum estimation problems}},\ }\href
  {https://doi.org/10.1088/0305-4470/39/40/014} {\bibfield  {journal} {\bibinfo
   {journal} {J. Phys. A: Math. Theor.}\ }\textbf {\bibinfo {volume} {39}},\
  \bibinfo {pages} {12489} (\bibinfo {year} {2006})}\BibitemShut {NoStop}%
\bibitem [{\citenamefont {Grover}\ and\ \citenamefont
  {Rudolph}(2002)}]{grovrud}%
  \BibitemOpen
  \bibfield  {author} {\bibinfo {author} {\bibfnamefont {L.}~\bibnamefont
  {Grover}}\ and\ \bibinfo {author} {\bibfnamefont {T.}~\bibnamefont
  {Rudolph}},\ }\href {https://arxiv.org/abs/quant-ph/0208112} {\bibinfo
  {title} {Creating superpositions that correspond to efficiently integrable
  probability distributions}} (\bibinfo {year} {2002}),\ \Eprint
  {https://arxiv.org/abs/quant-ph/0208112} {arXiv:quant-ph/0208112 [quant-ph]}
  \BibitemShut {NoStop}%
\bibitem [{\citenamefont {Kitaev}\ and\ \citenamefont {Webb}(2009)}]{kitaev}%
  \BibitemOpen
  \bibfield  {author} {\bibinfo {author} {\bibfnamefont {A.}~\bibnamefont
  {Kitaev}}\ and\ \bibinfo {author} {\bibfnamefont {W.~A.}\ \bibnamefont
  {Webb}},\ }\href@noop {} {\bibinfo {title} {Wavefunction preparation and
  resampling using a quantum computer}} (\bibinfo {year} {2009}),\ \Eprint
  {https://arxiv.org/abs/0801.0342} {0801.0342} \BibitemShut {NoStop}%
\end{thebibliography}%

\appendix
\section{Statistical theory for first excited state\label{sec:app1}}
The probability distribution for $q$-quadrature measurement in the first excited state ($n = 1$) of the quantum harmonic oscillator can be written in exponential family form \cite{lehmanntheory}:
\begin{equation*}
p(q;d) = h(q)\exp[d T(q)-B(d)]
\end{equation*}
where:
\begin{align*}
h(q) &= \frac{2q^2}{\sqrt{\pi}} \\
T(q) &= -q^2 \\
B(d) &= -\frac{3}{2}\log(d).
\end{align*}
Note that for larger $n$, we cannot similarly rewrite $\langle q=x\vert \psi_{n}(d)\rangle^{2}$ in exponential family form because the term $H^2_n(q\sqrt{d})$ does not factor as a product of $d$-dependent terms and $q$-dependent terms. 
For $n =1$, both the MOM and MLE are $\frac{3M}{2\sum_{i=1}^M q_i^2}$ (recall from Section \ref{sec:esttheor} that this equivalence does not hold for general $n$). Because for the probe state $\ket{\psi_{1}(d)}$ the probability distribution function for $q$-measurement is in the exponential family of distributions, it satisfies regularity conditions (see, e.g., Chapter 6 of \cite{lehmanntheory}) that imply that the MLE $\hat{d}$ is asymptotically unbiased and efficient. In this case, the MOM estimator equals the MLE and achieves the quantum Cram\'{e}r-Rao bound according to Theorem \ref{thm:ttt}. We next show that a Bayesian estimator can also be derived in analytical form.

\begin{theorem}Assume $M$ independent, identical $q$-measurements of $\ket{\psi_{1}(d)}$, referring to the collection of measurements  as $\mathbf{q}=(q_1,...,q_M)$. Assume a gamma prior $p(d) = \frac{\lambda^s}{\Gamma(s)}d^{s-1} \exp(-\lambda d)$, with shape $s$ and rate $\lambda$. Then the posterior distribution is also gamma distributed with shape $3M/2+s$ and rate $[\sum_{i=1}^M q^2_i]+\lambda$; i.e., the gamma distribution is a conjugate prior for $d$.
\end{theorem}

\begin{proof}The posterior distribution  $p(d | \mathbf{q})$ is proportional to the product of the prior and likelihood. Dropping constants not depending on $d$ yields  \begin{equation}p(d | \mathbf{q})\propto d^{3M/2+s-1} \exp(-\left(\sum_{i=1}^M q^2_i+\lambda\right) d),\end{equation} which can be recognized as a gamma distribution with stated shape and rate parameters.\end{proof}

As a consequence, the posterior mean is given by the expression $(3M/2+s)/(\sum_{i=1}^M q^2_i+\lambda)$. We note that from a decision-theoretic perspective, the posterior mean has favorable properties under a squared-error loss though for scale parameters other loss functions such as (\ref{eqn:covcost}) could be preferred. Regardless of the loss function chosen, the full posterior distribution is known exactly (i.e., including normalizing constant) under the gamma prior. 

Alternative priors motivated by the geometry of the parametrized state manifold can lead to exact agreement between the posterior mean and the MLE. For example, because the QFI is proportional to $1/d^2$ for any $n$, the Jeffrey's prior is proportional to $1/d$. Although the prior is improper since the integral of $1/d$ on $[0,\infty)$ does not converge, the posterior distribution is proportional to $d^{3M/2-1} \exp(-(\sum_{i=1}^M q^2_i) d)$, which is a proper gamma distribution for $M \geq 1$. In this case, the posterior mean is $\frac{3M}{2\sum_{i=1}^M q^2_i}$, which is equal to the MLE and MOM estimators.

\end{document}